\documentclass{article}
\usepackage{amsfonts}
\usepackage{amsmath}
\usepackage{amsthm}   
\usepackage{amssymb,amsmath}
\usepackage{setspace}

\title {Structure, classification, and conformal symmetry of elementary particles over non-archimedean space-time}

\font\mysmall=cmr8 at 8 pt
\font\eightit=cmti8

\def\c{\bf C}

\def\q{\bf Q}

\def\r{\bf R}

\def\hh{\cal H}
\def\kk{\cal K}

\newtheorem{theorem}{Theorem}[section]
\newtheorem{corollary}{Corollary}[section]
\newtheorem{lemma}{Lemma}[section]

\newtheorem{definition}{Definition}

\newtheorem{proposition}{Proposition}[section]

\newcommand{\Q}{\textbf{Q}}
\newcommand{\R}{\textbf{R}}
\newcommand{\C}{\textbf{C}}

\newcommand{\SO}{\textrm{SO}}
\newcommand{\SL}{\textrm{SL}}
\newcommand{\ORTH}{\textrm{O}}

\begin{document}

\maketitle

\centerline {\bf V. S. Varadarajan and Jukka T. Virtanen}

\setcounter{tocdepth}{1}
\tableofcontents

\section{Abstract}
{\mysmall
It is well known that at distances shorter than Planck length, no length measurements are possible.  The Volovich hypothesis asserts that at sub-Planckian  distances and times, spacetime itself has a non-Archimedean geometry.   We
discuss the structure of elementary particles, their classification,  and their  conformal symmetry under this hypothesis.  Specifically, we investigate the projective representations of the $p$-adic Poincar\'{e} and Galilean
groups, using a new variant of the Mackey machine for projective unitary representations of semidirect products of locally compact and
second countable (lcsc) groups. We construct the conformal spacetime over $p$-adic fields and discuss the imbedding of the $p$-adic Poincar\'{e} group into the $p$-adic conformal group.  Finally, we show that the massive and so
called eventually masssive particles of the Poincar\'{e} group do not have conformal symmetry. The whole picture bears a close resemblance to what happens over the field of real numbers, but with some significant variations.

\vskip 0.3 true in\noindent
Key words: Volovich hypothesis, non-archimedean fields, Poincar\'e group, Galilean group, semidirect product, cocycles, affine action, conformal spacetime, conformal symmetry, massive, eventually massive, and massless particles.
\vskip 0.3 true in\noindent
Mathematics Subject Classification 2000: 22E50, 22E70, 20C35, 81R05.}
\vskip 0.3 true in\noindent

\section{Introduction}
Divergences in quantum field theories led many physicists, notably Beltrametti and his
collaborators, to propose in the 1970's the idea that one should
include the structure of space-time itself as an unknown to be
investigated \cite{bel1} \cite{bel2}\cite{bel3} \cite{bel4}. In particular they suggested that the geometry of
space-time might be based on a non-archimedean or even a finite
field, and examined some of the consequences of this hypothesis. But
the idea did not really take off until Volovich proposed in 1987 \cite{vol}
that world geometry at sub-Planckian regimes might be
non-archimedean because no length measurements are possible at such
ultra-small distances and time scales. A huge number of articles
have appeared since then, exploring this theme. Since no single
prime can be given a distinguished status, it is even more natural to
see if one could really work with an adelic geometry as the basis
for space-time. Such an idea was first proposed by Manin \cite{man}. For  a
definitive survey and a very inclusive set of references concerning $p$-adic mathematical physics see the
article by Dragovich et al \cite{dra}. It is not our contention that there
is sufficient experimental evidence for a non-archimedean or adelic
spacetime. Rather we explore this question in the so-called {\it
Dirac mode\/}, namely to do the mathematics first {\it and then\/} to seek
the physical interpretation (see \cite{nam}) (p 371).

\vspace{14pt}

\noindent In this paper we examine the
consequences of the non-archimedean hypothesis for the
classification of elementary particles. We consider both the
Poincar\'e and the Galilean groups. Each of these is the group of
$k$-points of a linear algebraic group defined over a local
non-archimedean field $k$ of characteristic $\not=2$.

\vspace{14pt}

\noindent Beyond the classification of elementary particles with Poincar\'e and conformal symmetry lies the problem of constructing quantum field theories over $p$-adic spacetimes. For a deep study of this question see
the paper of Kochubei and Sait-Ametov \cite{koc}.

\vspace{14pt}

\noindent It is a consequence of the basic principles of quantum mechanics
(see \cite{var1}) that the symmetry of a quantum system with respect to a
group $G$ may be expressed by a projective unitary representation
(PUR) of $G$ (or at least of a normal subgroup of index $2$ in $G$) in the Hilbert space of quantum states; this PUR
may be lifted to an ordinary unitary representation (UR) of a
suitable topological central extension (TCE) of it by the
circle group $T$. Already in 1939, Wigner, in his great paper \cite{wig},
proved that all PUR's of the Poincar\'e group $P$ lift to UR's of
the simply connected (2-fold) covering group $P^\ast$ of the Poincar\'e
group.   In other words, $P^\ast=V \rtimes {\rm Spin}(V)$ is already the {\it universal\/} TCE of the Poincare
group (UTCE).  Here $V$ is a \emph{real}
quadratic vector space, namely a real vector space with a quadratic
form, of signature $(1, n)$ that defines the Minkowski metric, and
${\rm Spin}(V)$ is the spin group, which is the simply connected
covering group of the orthogonal group ${\rm SO}(V)$. We note that for $n=3$, Spin$(V)=\SL(2,\C)_\R$, 
the suffix $\R$ denoting the fact that we view $\SL(2,\C)$ as a real group.
For the real Galilean group, going to the simply connected covering group
is not enough to unitarize all PUR's.  One has to construct the UTCE
(see \cite{cas}.)

\vspace{14pt}

\noindent Not all groups have UTCE's.  For a lcsc group to have a UTCE it is necessary
that the commutator subgroup should be dense in it.  Over a non-archimedean local field, the commutator subgroups of the
Poincare group and the orthogonal groups are open and closed {\it
proper\/} subgroups and so they do not have UTCE's.
The spin groups and the Poincar\'e groups associated to the spin groups
{\it do have\/} UTCE's; for the spin groups this is a consequence of
the work of Moore \cite{moo} and Prasad and Raghunathan \cite{pra} and for the
corresponding Poincar\'e groups, of the work of Varadarajan \cite{var2}.
However, the natural map from the spin group or the corresponding
Poincar\'e group to the orthogonal group or the corresponding
Poincar\'e group is {\it not} surjective over the local
non-archimedean field $k$ (even though they are surjective over the algebraic closure of $k$), and so replacing the orthogonal group by
the spin group leads to a loss of information. So we have to work with the orthogonal group rather than the spin group.  The following example, treated in \cite{var0}, illustrates this.

\vspace{14pt}

\noindent Let $G={\rm SL}(2, {\q}_p)$. The adjoint representation exhibits $G$ as the spin group corresponding to the quadratic vector space $\frak g$ which is the Lie algebra of $G$ equipped with the Killing form. The adjoint map
$G\longrightarrow G_1={\rm SO}(\frak g)$ is the spin covering for ${\rm SO}(\frak g)$ but this is {\it not surjective\/}; in the standard basis
$$
X=
\left(
  \begin{array}{cc}
    0 & 1 \\
    0 & 0 \\
  \end{array}
\right)
,\quad H=
\left(
  \begin{array}{cc}
    1 & 0 \\
    0 & -1 \\
  \end{array}
\right), \quad Y=\left(
  \begin{array}{cc}
    0 & 0 \\
    1 & 0 \\
  \end{array}
\right)
$$
the spin covering map is
$$
\left(
  \begin{array}{cc}
    a & b \\
    c & d \\
  \end{array}
\right) \longmapsto
\left(
  \begin{array}{ccc}
    a^2 & -2ab & -b^2 \\
    -ac & ad+bc & bd \\
    -c^2 & 2cd & d^2 \\
  \end{array}
\right)
$$
The matrix
$$
\left(
  \begin{array}{ccc}
    \alpha & 0 & 0 \\
    0 & 1 & 0 \\
    0 & 0 & \alpha^{-1} \\
  \end{array}
\right)
$$
is in ${\rm SO}(\frak g)$; if it is the image of $\left(
                                                    \begin{array}{cc}
                                                      a & b \\
                                                      c & d \\
                                                    \end{array}
                                                  \right)$, then $b=c=0, d=a^{-1}$, and $\alpha=a^2$, so that unless $\alpha \in {{\q}_p^\times}^2$, this will not happen.

\vspace{14pt}

\noindent So in this paper we
work with the orthogonal groups rather than the spin groups. This means that we have to deal with projective UR's of the Poincar\'e and Galilean groups directly.
\vspace{14pt}

\noindent
An announcement containing the main results of this paper  (without proofs) has appeared in the Letters in Mathematical Physics
\cite{var0}
The present article is an elaboration of this announcement, with proofs.

\section{Multipliers and PURs for semidirect product groups}\label{section Multipliers and PURs for semidirect product groups}

\subsection{Multipliers for semidirect products}\label{section multipliers for semidirect products}
We assume that the reader is familiar with the basic facts regarding
multipliers \cite{var2}.  We begin by discussing the multiplier
group of a semidirect product. We work in the category of locally compact second countable (lcsc) groups.

\vspace{14pt}

\noindent Multipliers of a group $K$ form a group $Z^2(K)$. The subgroup of trivial multipliers is denoted by $B(K)$.  We define $H^2(K) = Z^2(K) \slash B(K)$ to be the \emph{multiplier group of $K$}. If $m\in Z^2(K)$ we define a
$m$-representation of $K$ to be a Borel map $x\longmapsto U(x)$
of $K$ into the unitary group of a Hilbert space ${\hh}$ (which is a standard Borel group) such that
$$
U(x)U(y)=m(x, y)U(xy)\qquad (x,y\in K).
$$
 If $K$ is totally disconnected, every multiplier is equivalent to a continuous one and the the subgroup of continuous multipliers has the property that the natural inclusion map induces an isomorphism with $H^2(K)$.
This is true for the $p$-adic groups.

\vspace{14pt}

\noindent Let $H = A \rtimes G$ where $A$ and $G$ are lcsc groups and $A$ is abelian. Let $A^*$ be the
character group of $A$.  Our starting point is to investigate the
subgroup of multipliers of $H$ that are trivial when restricted to $A$,
denoted by $M_A(H)$.  Let $H^2_A(H)$ be its image in
$H^2(H)$.  Let $M'_A(H)$ be the group of multipliers $m$ for $H$
with $m\vert_{A \times A} = m\vert_{A \times G} = 1.$ Results from \cite{var2} and \cite{mac1} tell us that any element of $M_A(H)$ is equivalent to one in $M'_A(H)$.
We define a $1$-cocycle for $G$ with coefficients in $A^*$ as a Borel
map $f(G \rightarrow A^*)$ such that $$f(gg') = f(g) + g [ f(g') ]
\;\; (g,g' \in G).$$ This is equivalent to saying that $g \mapsto
(f(g),g)$ is a Borel homomorphism of $G$ into the semidirect product
$A^* \rtimes G$.  Hence any $1$-cocycle is continuous and defines a continuous map of $G\times A$ into $T$. We denote the
abelian group of $1$-cocycles by $Z^1(G,A^*)$. The
coboundaries are the cocycles of the form $g \mapsto g [ a ] - a$
for some $a \in A^*$. The coboundaries form a subgroup $B^1(G,A^*)$
of $Z^1(G,A^*)$. We now form the cohomology group $H^1(G,A^*) =
Z^1(G,A^*) \slash B^1(G,A^*)$.   The following
theorem due to Mackey describes the multipliers of $H$. Full details can be found in \cite{var2}.

\begin{theorem}\label{theorem theorem 1 moscow paper} Any element in $M_A(H)$ is
equivalent to one in $M'_A(H)$. If $m \in M'_A(H)$ and $m_0 =
m\vert_{G \times G}$ and $\theta_m(g^{-1})(a') = m(g,a')$, then $m
\mapsto(m_0, \theta_m)$ is an isomorphism $M'_A(H) \simeq Z^2(G)
\times Z^1(G,A^*)$ which is well defined in cohomology and gives the
isomorphisms $H^2_A(H) \simeq H^2(G) \times H^1(G,A^*)$. Moreover,
$$
m(ag, a'g')=m_0(g, g')\theta_m(g^{-1})(a').
$$
\end{theorem}
\begin{corollary}If $m_0=1$, then $m$ is a continuous multiplier and $m(ag,a'g') =
\theta(g^{-1})(a')$.\end{corollary}
\noindent \textbf{Remark:} A multiplier $m$ for $H$ is said to be \emph{standard}
if $m\vert_{A\times A} = m\vert_{A \times G} = m\vert_{G \times A}
=1$.  It follows from the above that a multiplier for $H$ is standard
if and only if it is the lift to $H$ of a multiplier for $G$ via $H \rightarrow
H \slash A \simeq G$.

\subsection{$m$-Systems of imprimitivity}\label{section m-systems of
imprimitivity}

\noindent Classically, systems of imprimitivity are the key to finding UIR's of semidirect products.
In this section we utilize $m$-systems of imprimitivity to describe the PUIR's of semidirect product groups. This is a straightforward variation of the corresponding theory for ordinary systems of imprimitivity.

\vspace{14pt}

\noindent We assume the following setup.  Let $G$ be a lcsc group.
Let $X$ be a $G$-space that is also a standard Borel space. Let
$\mathcal{H}$ be a separable Hilbert space and $\mathcal{U}$ the
group of unitary transformations of $\mathcal{H}$. An
\emph{$m$-system of imprimitivity} is a pair $(U,P)$, where $P(E
\rightarrow P_E)$ is a projection valued measure (pvm) on the class of
Borel subsets of $X$, the projections being defined in
$\mathcal{H}$, and $U$ is an $m$-representation of $G$ in
$\mathcal{H}$  such that
$$U(g)P(E)U(g)^{-1} = P(g[E]) \;\; \forall \; g\in G \text{ and all
Borel } E \subset X.$$ The pair $(U,P)$ is said to be \emph{based on
$X$}.
For what follows we take $X$ to be a {\it transitive\/} $G$-space. We
fix some $x_0\in X$ and let $G_0$ be the stabilizer of $x_0$ in $G$,
so that $X \simeq G \slash G_0$. We will also fix a multiplier $m$
for $G$ and let $m_0 = m\vert_{G_0 \times G_0}$.
\begin{lemma}\label{lemma strict cocycle}
Let $G$ be a lcsc group acting transitively on a lcsc space $X$. Fix
$x_0 \in X$ and let $G_0$ be the stabilizer of $x_0$ in $G$. Suppose
that $C \subset G_0$ is a closed central subgroup of $G$ and $\chi$
is a character of $C$.  Let $\gamma$ be a strict $(G,X)$-cocycle
with the values in the unitary group $\mathcal{U}$ of Hilbert space
$\mathcal{K}$ and let $\nu$ be the map $G_0 \rightarrow \mathcal{U}$
defined by $\nu(g) = \gamma(g,x_0) \;\;\; g \in G_0$. Then $C$ acts
trivially on $X$ and the following are equivalent:
\begin{enumerate}
\item $\nu(c) = \chi(c)$.
\item $\gamma(c,x) = \chi(c)$ for each $c \in C$ for almost all
$x \in X$.
\item $\gamma(c,x) = \chi(c)$ for each $c \in C$, and for all $\; x
\in X$.
\end{enumerate}
\end{lemma}
\noindent The proof is a straightforward calculation and we omit it here for brevity.

\vspace{14pt}

\noindent Let $G$ be a lcsc group and $m$ a multiplier of $G$.  We recall now the Mackey technique of trivializing $m$ by passing to
a central extension of $G$. Let $T$ the
circle group of complex numbers. One
may build a central extension $E_m$ of $G$ given by the following
exact sequence:
$$0 \rightarrow T \rightarrow E_m \rightarrow G \rightarrow 0.$$
Here $E_m = G \times T$, the product structure is given by
$$(x_1,t_1)(x_2,t_2) = (x_1x_2, m(x_1x_2)t_1t_2),$$ and the maps are $t\mapsto (1,t)$ and $(x,t)\mapsto x$. The
Mackey-Weil topology on $E_m$ will convert $E_m$ into a lcsc group.
The key property of $E_m$ is that any $m$-representation of $G$
lifts to a unitary representation on $E_m$ that restricts to
$(1,t)$ as $tI$.

\vspace{14pt}

\noindent Let $E= E_m$ be as above and $E_0$ the
preimage of $G_0$ in $E$ under the map $E \rightarrow G$.  Then $E =
G \times T$ and $E_0 = G_0 \times T$. $E_0$ is isomorphic to
$E_{m_0}$, the central extension of $G_0$ defined by $m_0$.  We note
that $E$ acts on $X$ through $G$, and so $E_0$ may be viewed as the
stabilizer of $x_0$ in $E$.
\begin{lemma}\label{lemma correspondance imprimitivity URs of E0}
In the correspondence between systems of imprimitivity $(V,Q)$ for
$E$ based on $X$ and UR's $\nu$ of $E_0$, the systems with $V(t) = t
I \;\; (t \in T)$ correspond to UR's $\nu$ of $E_0$ with $\nu(t) = t
 I$.
\end{lemma}
\begin{proof}
Let $\nu$ and $V$ correspond. Then $\nu$ gives rise to a strict
$(E_0,X)$-cocycle  $\gamma$ such that $\gamma(g,x_0) = \nu(g)$.  The
representation $V$ acts on the Hilbert space $\mathcal{H} =
L^2(X,\mathcal{K},\lambda)$ of $\mathcal{K}$-valued functions on $X$ and
$\lambda$ is a quasi-invariant probability measure on $X$. The action of
$V$ is given by
$$(V(h)f)(x) = \rho_h(h^{-1} [x])^{\frac{1}{2}} \gamma(h,h^{-1}
[x]) f(h^{-1}  [x])$$ where $\rho_h = d\lambda\slash d\lambda^{h^{-1}}$.
Note that $T$ is central in
$E$ and so acts trivially on $X$ by Lemma 3.1, and therefore $\rho_t = 1$ for all $t\in T$. Suppose now that $\nu(t) = tI$ for $t \in T$.  By Lemma \ref{lemma
strict cocycle}, $\gamma(t,x) = t$ for all $x \in X$. This shows that $V(t) =
tI$. Conversely, suppose that $V(t) = tI$ for all $t \in T$. Then
for each $t \in T$, $\gamma(t,t^{-1}[x]) = t$ for almost all $x$ so
that $\gamma(t,x) = t$ for almost all $x$.  By Lemma \ref{lemma
strict cocycle} we have that $\nu(t) = t$ for all $t \in T$.
\end{proof}
\begin{theorem}\label{theorem correspondance m0 reps and msystems}
There is a natural one to one correspondence between
$m_0$-representations $\mu$ of $G_0$ and $m$-systems of
imprimitivity $S_\mu:=(U,P)$ of $G$  based on $X$. Under this
correspondence, we have a ring isomorphism of the commuting ring of
$\mu$ with that of $S_\mu$, so that irreducible $\mu$ correspond to irreducible $S_\mu$.
\end{theorem}
\begin{proof}
If $(U,P)$ is an $m$-system of imprimitivity for $G$, and we define
$V$ on $E_m$ by $V(x,t) = tU(x)$, then $V$ is an ordinary
representation and $(V,P)$ is thus an ordinary system of
imprimitivity for $E_m$ based on $X$.  We have $V(t)=tI$ for all $t\in T$ and all such $(V, P)$ arise in this manner. Lemma \ref{lemma
correspondance imprimitivity URs of E0} now says that there is a
bijection between these $(V,P)$  and UR's
$\nu$ of $E_0$ such that $\nu(t) = tI$. Since $E_0 \simeq
E_{m_0}$, there is a bijection between $\nu$-representations of
$E_0$ for which $\nu(t) = tI$ and $m_0$-representations $\mu(x) =
\nu(x,1)$ of $G_0$.  So, there exists a bijection between
$m$-systems of imprimitivity $(U,P)$ for $G$ and
$m_0$-representations of $G_0$.  The commuting rings of $\nu$ and
$U$ are respectively the same as the commuting rings of $\mu$ and
$V$, which are isomorphic by the Mackey theory.
\end{proof}
\noindent We will now make the above correspondence more explicit.  We define
\emph{a strict $(G,X)$-$m$-cocycle} to be a Borel map $\delta:G
\times X \rightarrow \mathcal{U}$ such that:
$$m(g_1,g_2)\delta(g_1g_2,x) = \delta(g_1,g_2 [ x ])
\delta(g_2,x) \;\;\;\; \forall g_i \in G, \;\; x \in X.$$ Two such
cocycles $\delta_i \;\; (i=1,2)$ are cohomologous $(\simeq)$ if
there exists a Borel function $\phi: X \rightarrow \mathcal{U}$ such
that $\delta_2(g,x) = \phi(g [x])\delta_1(g,x)\phi(x)^{-1}$ for all
$g \in G, \;\; x\in X$.
\begin{lemma}\label{lemma bijection between strict cocycles and strict m cocycles} There is a natural bijection
between strict $(E,X)$-cocycles $\gamma$ such that $\gamma((t),x)
= tI \;\; \forall t \in T, \; x \in X$, and strict
$(G,X)$-$m$-cocyles $\delta$, which given by $\delta(g,x) =
\gamma((g,1),x)$, $\gamma((g,t),x) = t \delta(g,x)$.  This bijection
respects equivalences, and induces a bijection of the respective
cohomology sets.
\end{lemma}
\begin{proof}

\noindent We have,
$$\gamma((g_1,t_1)(g_2,t_2),x) =
\gamma((g_1g_2,t_1t_2m(g_1g_2)),x)$$
$$= t_1t_2m(g_1,g_2)\delta(g_1g_2,x) = t_1t_2\delta(g_1,g_2 [x])\delta(g_2 ,x) $$
$$ = \gamma((g_1,t_1),(g_2,t_2) [x]) \gamma ((g_2,t_2),x).$$
Also,
$$\delta(g_1g_2,x) = \gamma((g_1g_2,1),x) = \gamma(g_1,1)(g_2,
m(g_1,g_2)^{-1},x)$$
$$=\gamma((g_1,1),g_2 [x]) \gamma((g_2,m(g_1,g_2)),x)$$
$$\delta(g_1,g_2 [x]) m(g_1,g_2)^{-1} \delta(g_2,x).$$
The fact that the correspondence respects equivalence is clear since
$X$ is the same for both and since the condition $\gamma((t),x)=t
I \; \forall \; t\in T, \; x\in X$ is unchanged under equivalence.
\end{proof}
\begin{lemma}\label{lemma m-representation Hilbert space action}
Given an $m_0$-representation $\mu$ of $G_0$ there is a strict
$m$-cocycle $\delta$ with values in $\mathcal{U}$ such that
$\delta(g,x_0) = \mu(g)$.  In the corresponding $m$-system $(U,P)$,
the action of $U$ is given as follows: $U$ acts on
$L^{2}(X,\mathcal{K},\lambda)$ with
$$(U(g)f)(x) = \rho_g(g^{-1} [x])^{\frac{1}{2}}\delta(g,g^{-1}
[x])f(g^{-1} [x]).$$ The $\rho$ factors drop out if $\lambda$ is
invariant.
\end{lemma}
\begin{proof}
Define $\nu(x,t)=t \mu(x)$ for $(x,t) \in E_0$.  Then $\nu$ is a UR
of $E_0$. This is a consequence of the fact that $E_0 \simeq
E_{m_0}$.  Next, we build a strict $(E,X)$-cocyle $\gamma$ with
$\gamma((g,t),x_0) = \nu((g,t)), \; (g,t) \in E_0$.  By Lemma
\ref{lemma bijection between strict cocycles and strict m cocycles},
such strict $(E,X)$-cocycles are in bijection with strict
$(G,X)$-$m$-cocyles $\delta$ given by $\delta(g,x) =
\gamma((g,1),x)$ and $\gamma((g,t),x) = t \delta(g,x)$.
\end{proof}
\noindent We will need the following lemma for later use.
\begin{lemma}\label{lemma strict cocycle lemma} Let $\delta_i$ $(i=1,2)$ be two strict $m$-cocycles
for $G$ such that for each $g \in G$, $\delta_1(g,x)=\delta_2(g,x)$
for almost all $x \in X$.  Let $\nu_i$ be the $m$-representations of
$G_0$ defined by $\delta_i$ $(i=1,2)$.  Then $\nu_1 \simeq \nu_2$.
\end{lemma}
\begin{proof}
Let $\gamma_i$ be the strict $(E,X)$ cocycle defined by
$\gamma_i((g,t),x) = t \delta_i(g,x).$  Then for each $(g,t) \in E$,
$\gamma_1((g,t),x) = \gamma_2((g,t),x)$ for almost all $x \in X$.
Let $\mu_i((g,t)) = \gamma_i((g,t),x_0),  \;\; g\in G_0$.  Then by
\cite{var1} p. 178 Lemma 5.25, $\mu_1 \simeq \mu_2$.  Since
$\nu_i(g) = \mu_i(g,1))$, $\nu_1 \simeq \nu_2$.
\end{proof}
\subsection{The Mackey machine for projective unitary irreducible representations of semidirect products}\label{section The Mackey machine for projective representations}
We now turn our attention to the Mackey treatment
of lcsc groups with a semidirect  product structure. In this section
we are going to consider a group $H= A \rtimes G$ where $G$ and $A$
are lcsc groups and $A$ is abelian. We concern ourselves only
with multipliers of $H$ which are trivial when restricted to $A
\times A$.  We recall that these multipliers are completely
described by Theorem \ref{theorem theorem 1 moscow paper}.

\vspace{14pt}

\noindent The following lemma introduces the key idea  that is needed for the variant of the Mackey machine for semidirect products when we deal with projective unitary representations.
\begin{lemma}\label{lemma affine action} Let $\phi:G \rightarrow A^*$ be a continuous map with $\phi(1)=0$.
Define $g\{\chi\}=g_\phi\{\chi \} = g[\chi] + \phi(g)$, for $g \in G, \chi \in
A^*$.  Then $a_{\phi}:(g, \chi) \mapsto g_\phi\{\chi\}$ defines an action of $G$ on
$A^*$ if and only if $\phi \in Z^1(G,A^*)$.\end{lemma}
\begin{proof}
\noindent If $a_\phi$ is to be an action on $A^*$, then $g_2\{g_1\{\chi
\} \} = g_2g_1 \{ \chi \}$ for all $\chi \in A^*$. Now
$$g_2\{g_1\{\chi \} \} = g_2[g_1[\chi] + \phi(g_1)] + \phi(g_2) = g_2[g_1[\chi]] + g_2[\phi(g_1)] + \phi(g_2).$$
On the other hand
$$g_2g_1\{\chi\} = g_2g_1[\chi] + \phi(g_2g_1).$$
Equating the two we see that the condition on $\phi$ is
$$\phi(g_2g_1) = g_2[\phi(g_1)] + \phi(g_2),$$
that is, $\phi \in Z^1(G,A^*).$
\end{proof}
If $\phi' \in Z^1(G,A^*)$ defines the same element as $\phi$ in $H^1(G,A^*)$, then
$\phi'(g) = \phi(g) + g [ \chi_0] - \chi_0$ for some $\chi_0 \in
A^*$.   So,
$$g_{\phi'} \{\chi\} = g[\chi] + \phi'(g) = g[\chi] + g [ \chi_0] -
\chi_0 = g[\chi + \chi_0] - \chi_0.$$ Let $\tau: \chi \mapsto \chi +
\chi_0$ be the translation by $\chi_0$ in $A^*$.  Then,
$$g_{\phi'} = \tau^{-1} g_{\phi} \tau.$$
So the actions defined by $\phi$ and $\phi'$ are equivalent in this
strong sense.
\begin{definition} The action $a_\phi: (g,\chi) \mapsto g_\phi\{\chi\}$ is called
the affine action of $G$ on $A^*$ determined by $\phi$.
\end{definition}
\begin{theorem}\label{theorem mackey theorem for m-systems of imprimitivity}Fix $\theta \in Z^1(G,A^*)$ and $m \in M_A'(H), \; m \simeq (m_0,\theta)$.
Then there is a natural bijection between $m$-representations $V$ of $H=A
\rtimes G$ and $m_0$-systems of imprimitivity $(U,P)$ on $A^*$ for
the affine action $g, \chi \mapsto g_\theta\{\chi\}$, defined by $\theta$. The bijection is given by:
$$V(ag) = U(a)U(g), \;\;\; U(a) = \int_{A^*}\langle a,\chi\rangle  dP(\chi).$$
\end{theorem}
\begin{proof}
The assumption $m\simeq (m_0, \theta)$ means that
$$m(ag,a'g') = m_0(g,g')\theta(g^{-1})(a')$$
where $m_0$ is a multiplier for $G$ and $\theta$ is a cocycle in
$Z^1(G,A^*)$. Let $V$ be a $m$-representation for $H$ and let us write $U$ for the restriction of $V$ to $A$ and $G$. Then $U$ is an ordinary representation of $A$ as well as a $m_0$-representation of $G$ and
$$
V(ag)=U(a)U(g).
$$
Moreover
$$U(g)U(a)U(g)^{-1} = \theta(g^{-1})(a)U(g[a]).$$
Indeed, we have $\theta(g^{-1})(a)=m(g,a)$ and
$$
U(g)U(a)=m(g,a)V(ga)=m(g,a)V(g[a]g)=m(g,a)U(g[a])U(g).
$$
Since $U$ is an ordinary representation on $A$, there exists a unique pvm (projection valued measure) $P$ on
$A^*$ such that:
$$U(a) = \int_{A^*} \langle a, \chi\rangle  dP(\chi) \;\;\;\; (a \in A).$$
\noindent Thus
$$U(g) U(a) U(g)^{-1} =\int_{A^*}\langle a,\chi\rangle dQ_g(\chi).$$
Here $Q_g$ is the pvm defined by $Q_g(E) = U(g)P(E)U(g)^{-1}$. On the other hand
$$
U(g)U(a)U(g)^{-1}=\theta(g^{-1}(a)U(g[a])
=\theta(g^{-1})(a)\int_{A^\ast}\langle g[a],\chi\rangle dP(\chi)
$$
and the right side can be rewritten as
$$
\int_{A^\ast}\langle a,g^{-1}[\chi]+\theta(g^{-1})\rangle dP(\chi)
=\int_{A^\ast}\langle a, g^{-1}\{\chi\}\rangle dP(\chi)
$$
so that
$$U(g)U(a)U(g)^{-1}=\int_{A^\ast}\langle a, g^{-1}\{\chi\}\rangle dP(\chi).
$$
Now, if $t$ is a Borel automorphism of $A^\ast$ as a Borel space and $f$ is a bounded Borel function on $A^\ast$, then
$$
\int_{A^\ast}f(t^{-1}(\chi)) dP(\chi)=\int_{A^\ast}f(\chi)dP_t(\chi)\eqno (\ast)
$$
where $P_t$ is the pvm defined by
$$
P_t(E)=P(t[E]).
$$
To see this, observe that $(\ast)$ is true if $f=1_E$, the characteristic function of a Borel set $E\subset A^\ast$; hence $(\ast)$ is true for $f$ which are finite linear combinations of such characteristic functions, and hence
also for all their uniform limits, which are precisely all bounded Borel functions. Hence we get
$$
U(g)U(a)U(g)^{-1}=\int_{A^\ast}\langle a, g^{-1}\{\chi\}\rangle dP(\chi)=
\int_{A^\ast}\langle a, \chi\rangle dP_g(\chi)
$$
where $P_g$ is the pvm defined by
$$
P_g(E)=P(g\{E\}).
$$
But we had seen that
$$U(g) U(a) U(g)^{-1} =\int_{A^*}\langle a,\chi\rangle dQ_g(\chi).$$
Hence
$$
\int_{A^\ast}\langle a, \chi\rangle dQ_g(\chi)=
\int_{A^\ast}\langle a, \chi\rangle dP_g(\chi)
$$
showing that
$$
Q_g=P_g
$$
or
$$
U(g)P(E)U(g)^{-1}=P(g\{E\}).
$$
We have thus shown that for the action of $G$ on $A^*$ by $g,\chi \mapsto
g \{ \chi \}$, $(U,P)$ is an $m_0$-system of imprimitivity.
Conversely, suppose $(U,P)$ is an $m_0$-system of imprimitivity for
this action.  Then, by retracing the steps in the above calculation
with $U(a) = \int_{A^*} \langle a,\chi\rangle  dP(\chi)$ we find:
$$U(g)U(a)U(g)^{-1} = \theta(g^{-1})(a)U(g[a]).$$
If we define $V(ag) = U(a)U(g)$, then $V$ becomes an
$m$-representation where $m(ag,a'g') = m_0(g,g')\theta(g^{-1})(a')$.
\end{proof}
\noindent We need the following definition.
\begin{definition} If $U$ is a UR of $A$ and $E \mapsto P(E)$ is its
associated pvm, we say that $Spec(U) \subset F$, if $P(F) =I$.  Here
$F$ is some Borel set in $A^*$.  We extend this terminology to any
PUR of $H=A \rtimes G$ that is a UR on $A$.
\end{definition}
\noindent By combining theorems \ref{theorem correspondance m0 reps and msystems}, \ref{theorem mackey theorem for m-systems of imprimitivity} and Lemma \ref{lemma m-representation Hilbert space action} we obtain the basic theorem
of irreducible $m$-representations
of $H$.
\begin{theorem}\label{theorem Mackeys theorem for m-reps} Fix $\chi \in A^*$, $m \simeq(m_0, \theta)$.  Then there is a natural
bijection between irreducible $m$-representations $V$ of $H=A
\rtimes G$ with $Spec(V) \subset G\{\chi \}$ (the orbit of $\chi$
under the affine action) and irreducible $m_0$-representations of
$G_\chi$, the stabilizer of $\chi$ in $G$ for the affine action.  If
the affine action is regular, every irreducible $m$-representation
of $H$, up to unitary equivalence, is obtained by this procedure.
Let $X =G\{\chi\}$ and $\lambda$ be a $\sigma$-finite quasi-invariant
measure for the action of $G$. Then, for any irreducible
$m_0$-representation $\mu$ of $G_\chi$ in the Hilbert space
$\mathcal{K}$, the corresponding $m$-representation $V$ acts on
$L^2(X,\mathcal{K},\lambda)$ and has the following form:
$$(V(ag)f)(\chi) = \langle a,\chi\rangle  \rho_g (g^{-1} \{ \chi \} )^{\frac{1}{2}})
\delta(g,g^{-1} \{ \chi \})f(g^{-1} \{ \chi \} )$$
where $\delta$ is any strict $m_0$-cocyle for $(G,X)$ with values in $\mathcal{U}$,
the unitary group of $\mathcal{K}$, such that $\delta(g,\chi) =
\mu(g), \; g\in G_\chi$.
\end{theorem}
\noindent We note that the $\rho_g$-factors drop out if $\lambda$ is an invariant
measure.
\begin{corollary}Suppose $H^1(G,A^*)=0$. Then we can take
$\theta(g)=1$ so that $m(ag,a'g')= m_0(g,g')$.  In this case, the
affine action reduces to the ordinary action.
\end{corollary}
\section{PUIR's of the $p$-adic Poincar\'{e} group and particle classification }\label{section Particle classification of the Poincare group}
\subsection{Preliminaries}
We shall now discuss the PUIRs of the $p$-adic Poincar\'{e} group.  For
all that follows we work over the field $\Q_p$. All the groups
described will be algebraic groups defined over $\Q_p$, so that the
groups of their $\Q_p$-points are $p$-adic Lie groups, in particular
lcsc. By the Poincar\'{e} group  we mean the group $P_V = V
\rtimes \SO(V)$ where $V$ is a finite-dimensional quadratic vector
space over $\Q_p$. Elementary particles correspond to PUIRs of the
Poincar\'{e} group. Our aim is to describe these PUIRs and thus classify
the elementary particles associated to the $p$-adic Poincar\'{e}
group.

\vspace{14pt}

\noindent We want to establish first that the PUIRs of the
Poincar\'{e} group are indeed described by Theorem \ref{theorem
Mackeys theorem for m-reps}. This means we must establish that
$P_V$ satisfies the criteria required for Theorem \ref{theorem
Mackeys theorem for m-reps}. We shall replace $V^\ast$ by the algebraic dual $V'$ of $V$ since $V^*$, the topological dual of $V$, is isomorphic to the algebraic dual $V'$, the isomorphism being natural and compatible with actions
of ${\rm GL}(V)$.  See for example \cite{weil}. The isomorphism is easy to set up but depends on the choice of a non-trivial additive character on ${\q}_p$, say $\psi$. Once we choose $\psi$, then, for any $p\in V'$, $
\chi_p : a\longmapsto \psi (\langle a, p\rangle)$ is in $V^\ast$, and $p\longmapsto \chi_p$ is a topological group isomorphism of $V'$ with $V^\ast$.

\vspace{14pt}

\noindent We note that the cohomology $H^1(\SO(V),V')$ is trivial. This is because
$\SO(V)$ is semisimple, and so, by theorem 3 of \cite{var2}
$H^1(\SO(V),V')=0$.  Hence, by earlier remark, every multiplier of $P_V$ is equivalent to the lift of a multiplier
of $\SO(V)$.

\vspace{14pt}

\noindent Theorem \ref{theorem Mackeys theorem for m-reps}
requires that the action of $\SO(V)$ on $V'$ is regular.
Since the quadratic form on $V$ is nondegenerate we have,
canonically, $V \simeq V'$.  We transfer the quadratic form in $V$
to $V'$, denoting it again by $( \cdot, \cdot)$.  The action of ${\rm SO}(V)$ on $V^\ast$ then goes over to the action of ${\rm SO}(V)$ on $V'$. Since the quadratic form on $V'$ is invariant under ${\rm SO}(V)$, the level sets of
the quadratic form are invariant sets. Under $\SO(V')$, $V'$ decomposes into invariant sets of the following types.
\begin{enumerate}
\item The sets $M_a = \{p \in V' \; \vert \; (p,p) = a \neq 0 \}.$
\item The set  $M_0 = \{p \in V' \; \vert \; (p,p)=0, \;p \neq 0 \}.$
\item The set  $\{ 0 \}$.
\end{enumerate}
We may think of the elements of $V'$ as momenta although this is
just formal.

\vspace{14pt}

\begin{lemma}\label{lemma orbits} If $\dim (V)\ge 3$, the sets $M_a$ and $\{0\}$ are all the orbits. Moreover the action is regular.\end{lemma}
\begin{proof}
First take $a\not=0$ and $p, p'\in M_a$. The map that takes $p$ to $p'$ is an isometry between their one dimensional spans, and so, we can extend it to an isometry $t$ of $V$ with itself. If $\det (t)=1$ we are done as $t\in {\rm
SO}(V)$. Suppose $\det (t)=-1$. If we can find an isometry $s$ fixing $p$ with $\det (s)=-1$, then $u=ts$ will be in ${\rm SO}(V)$ and take $p$ to $p'$. To see that we can find such an $s$, notice that for $U=p^\perp$, we have
$V=U\oplus \langle p\rangle$; as $\dim (U)\ge 1$ we can find $s'\in {\rm O}(U)$ with $\det (s')=-1$. Then $s$ can be defined as $s'$ on $U$ and $sp=p$, and we are done. Let $a=0$ and $p, p'\in M_0$. The argument is the same as
before and we are reduced to finding  $s$ as before. We can find $q\in V$ such that $(q, q)=0$ and $(p, q)=1$. Let $W$ be the span of $p$ and $q$. Then the quadratic form of $V$ is non-degenerate when restricted to $W$ and so
$V=W\oplus W^\perp$. We have $\dim (W^\perp)\ge 1$ and so we can find $s'\in {\rm O}(W^\perp)$ with $\det (s')=-1$. Then $s$ is defined as $s'$ on $W^\perp$ and identity on $W$, and we are done. Since $\{0\}$ is trivially an
orbit, we are finished. The regularity follows from the theorem of Effros \cite{eff}
as all orbits are obviously either closed or locally closed.
\end{proof}
\begin{definition}
We will call the orbits $M_a (a\not=0)$ massive, the orbit
$M_0$ massless, and $\{0\}$ trivial-massless.
\end{definition}

\vspace{14pt}

\noindent Next, Theorem \ref{theorem Mackeys theorem for m-reps} requires
that the multipliers of $P_V$ be trivial when restricted to $V$. To see this we use Corollary 2 to Proposition 2 of Section 4 of \cite{var2} and reduce the proof to showing that $0$ is the only skew symmetric invariant bilinear form on
$V$. But $V$ is irreducible under the ${\rm SO}(V)$ and admits a symmetric invariant bilinear form, namely $({\cdot}, {\cdot})$. Hence, any invariant bilinear form must be a multiple of this, and so, a skew symmetric invariant
bilinear form must be $0$.
\par
Finally, we shall show that all the orbits admit invariant measures.
\begin{lemma}\label{invariant measures} For $V$ of any dimension $\ge 1$, all the orbits of ${\rm SO}(V)$ admit invariant measures.
\end {lemma}
\begin{proof} Let $G$ be a unimodular lcsc group, and $H$ is a closed subgroup of $G$; then for $G/H$ to admit a $G$-invariant measure it is well known that the unimodularity of $H$ is a sufficient condition. We apply this to our
present situation. For $p\in V$ let $L_p$ be its stabilizer. We shall check that $L_p$ is unimodular for all $p$. We also check that the Poincar\'e group is unimodular, as it is needed in the proof.

\vspace{14pt}

\noindent {\it Poincar\'e group.\/} Here $P=V\rtimes G$ where $G={\rm SO}(V)$. The group $G$ acts on $V$ with determinant $1$ and so the action of the corresponding $p$-adic group $G_p$ on $V_p$ preserves Haar measure on $V_p$. It
is then easy to see that the product measure $dvdg$ is invariant under both left and right translations of $P$, $dv, dg$ being the respective Haar measures on $V, G$, provided we know that $G_p$ is unimodular. If $\dim (V)=1$,
$G=\{e\}$ and there is nothing to prove. If $\dim (V)=2$, then $G$ is abelian and so $G_p$ is unimodular. Let $\dim (V)\ge 3$. Then $G$ is semisimple. For $G_p$ to be unimodular it is enough to  check that its action on its Lie
algebra has determinant with $p$-adic absolute value $1$. Actually its determinant itself is $1$. It is enough to verify this last statement at the level of the algebraic closure of ${\q}_p$, where it follows from the fact that
over the algebraic closure $G$ is its own commutator group and so any morphism into an abelian algebraic group is trivial.

\vspace{14pt}

\noindent {\it The stabilizer of a massive point.\/} First consider $p\in M_a,\; a\not=0$. Then as we saw in the proof of the previous lemma, $V=U\oplus \langle p \rangle$, and $s\in {\rm SO}(V)$ fixes $p$ if and only in it leaves $U$
invariant and restricts to an element of ${\rm SO}(U)$ on $U$. Hence $L_p\simeq {\rm SO}(U)$, hence unimodular as observed above.

\vspace{14pt}

\noindent {\it Stabilizer of a massless point.\/} Let $p\in M_0$. We shall show in Theorem \ref{theorem imbedding} that $L_p\simeq P_W$ where $P_W$ is the Poincar\'e group of a quadratic vector space $W$ (with $\dim(W)=\dim(V)-2$ and $W$ is Witt
equivalent to $V$). Hence $P_W$ is unimodular from above.
\end{proof}

\vspace{14pt}

\noindent We now see that the Theorem \ref{theorem Mackeys theorem for m-reps}
applies to the $p$-adic Poincar\'{e} group and we summarize our results in
the following theorem that completely describes the particles of
the $p$-adic Poincar\'{e} group.  Recall that every multiplier for
$P_V$ is the lift to $P_V$ of a multiplier for $\SO(V)$, up to
equivalence. For any $p\in V$ we denote by $\lambda_p$ an invariant measure on the orbit of $V$. If $m_0$ is a multiplier for ${\rm SO}(V)$ and $m$ its lift to $P_V$, we write $m_p$ for the restriction of $m$ to the stabilizer of
$p$ in ${\rm SO}(V)$.
\begin{theorem}\label{theorem Particles of the Poincare}
Let $P_V=V \rtimes \SO(V)$ be the $p$-adic Poincar\'{e} group.  Fix
$p \in V'$ and let $m_0$ be a multiplier of $\SO(V)$ and $m$ its
lift to $P_V$. Then there is a natural bijection between irreducible
$m$-representations of $P_V=V \rtimes \SO(V)$ with $Spec(V) \subset
\SO(V) [p]$, the orbit of $p$ under the ordinary action of
$\SO(V)$ and irreducible $m_p$-representations of $\SO(V)_p$, the
stabilizer of $p$ in $\SO(V)$.  Every PUIR of $P_V$, up to unitary
equivalence, is obtained by this procedure. Let $X =\SO(V) [p]$,
$\lambda_p$ a $\sigma$-finite invariant measure on $X$ for the action of
$\SO(V)$. Then, for any irreducible $m_p$-representation $\mu$ of
$\SO(V)_p$ in the Hilbert space $\mathcal{K}$, the corresponding
$m$-representation $U$ acts on $L^2(X,\mathcal{K},\lambda_p)$ and has the
following form:
$$(U(ag)f)(p) = \psi (\langle a,p\rangle )
\delta(g,g^{-1} \{p\})f(g^{-1} \{p\} )$$
where $\delta$ is any strict $m_p$-cocyle for $(\SO(V)),X)$ with values in
$\mathcal{U}$, the unitary group of $\mathcal{K}$, such that
$\delta(g,p) = \mu(g), \; g\in \SO(V)_p$.
\end{theorem}

\vspace{14pt}

\noindent Thus, to determine the PUIRs of the Poincar\'{e} group, one must
determine the multipliers of $L=\SO(V)$ and for each given multiplier
$m$, determine the irreducible $m$-representations.  The PUIRs then
correspond to $p \in V'$ and $m_p$-representations of the
stabilizer $L_p$ of $p$ in $L$, $m_p$ being
$m\vert_{L_p \times L_p}$.

\vspace{14pt}

\noindent
We now define massless and massive particles.
\begin{definition} A PUIR of the Poincar\'{e} group is called an elementary particle.  A particle which corresponds to the orbit of a vector $p\in V'$ is called massless if $p\not=0$ and is massless $((p, p)=0)$, trivial if $p=0$,
and massive if $p$ is massive $((p, p) \not=0)$.
\end{definition}
\section{Galilean group}\label{section galilean group}
\subsection{Galilean group over $\R$}
Classically, the Galilean group is the group of translations,
rotations, and boosts, of spacetime consistent with  Newtonian
mechanics. Let $V_0={\r}^3$ be space with $x=(x_1,x_2,x_3)$ as space coordinates and $V_1={\r}$ be time with $t$ as time coordinate.  We define spacetime as $V=V_0
\oplus V_1$, and we write for $w \in V$, $w=(x,t)$.  Then a Galilean
transformation $g:w=(x,t) \mapsto w'=(x',t')$ is defined by
$$g:w\mapsto w'\qquad x'= Wx+tv+u, \quad t'=t+\eta.$$
Here $W \in \SO(3)$, $u$ and $v$ are vectors in $3$-space and $\eta$
is a real number.  In this transformation $u$ is a spatial
translation, $\eta$ is a time translation, $v$ is a boost.  We may
think of $v$ as a velocity vector and and $W$ a rotation in the
$3$-space. The set of all such transformations forms the Galilean
group. The Galilean group is a semidirect product $V \rtimes R$ of
the group $V$ of all translations in spacetime and the group $R =
V_0 \rtimes R_0$. Here $R_0= \SO(V_0)$. The subgroup $R$ is not
semisimple. This creates some subtle differences between the theory
involving the Poincar\'{e} group and the theory involving the
Galilean group \cite{var1} p. 283-284.
\subsection{Galilean group over $\Q_p$}
We define the analogue of the Galilean group over $\Q_p$.  Let
$V$ be a finite-dimensional vector space over $\Q_p$ such that $V =
V_0 \oplus V_1$ where $V_0$ is an isotropic quadratic vector space and $V_1$
has dimension $1$, which we identify with ${\q}_p$.  The Galilean group is now defined as $G = V
\rtimes R$ where $R = V_0 \rtimes \SO(V_0)$.  Technically one should
think of this as a pseudo-Galilean group since in the real case
$V_0$ is anisotropic. We need the presence of isotropic vectors in $V_0$ as a technical requirement that we cannot do away with. As before, the action of $((u,\eta),(v,W)) \in
G$ on $V$ is given by
$$((u,\eta),(v,W)):(x,t) \mapsto (Wx+tv+u,t+\eta)$$
Let $(\cdot , \cdot)$ be the bilinear form on $V_0$ that describes its quadratic structure.  Given a pair
$(\xi, t) \in V$, we define a linear form $\langle (\xi,t), \cdot \rangle $
on $V$ by $\langle (\xi,t), (u , \eta)\rangle  = (\xi, u) + t\eta$.  We
identify the algebraic dual $V'$ with set of all such pairs
$(\xi,t)$. We now describe the action of $R$ on $V$.
$$(v,W):(u,\eta) \mapsto (Wu + \eta v, \eta)$$
The action of $R$ on $V'$ is given by
$$(v,W):(\xi,t) \mapsto (W\xi, t-(W\xi,v))$$
\subsection{Particle classification of the $p$-adic Galilean
group}The study of particles of the real Galilean group corresponds to
the study of particles of ordinary non-relativistic quantum
mechanics. A natural question that arises is that of classifying
particles of the $p$-adic Galilean group. This classification is
a consequence of Theorem \ref{theorem Mackeys theorem for
m-reps}.  It is noteworthy that in the presence of a nontrivial
affine action the theorem differs from the usual Mackey theorem.
\subsection{Multipliers of the Galilean group}
To determine $H^2(G)$ we must first show that the multipliers of $G$ are trivial when restricted to $V$. This reduces to showing that $0$ is the only $R$-invariant skew symmetric bilinear form on $V$ (Corollary 2 to Proposition 2 of Section 4 of \cite{var2}).
Let $B$ be one such. As $V_0$ is invariant under $R$ with the action $(v, W), x\mapsto Wx$, the restriction $B_0$ of $B$ to $V_0$ is also skew symmetric and $R_0$-invariant. But $V_0$ already has a $R_0$-invariant {\it
symmetric\/} form, namely $({\cdot},{\cdot})$, which is non-degenerate. Since $V_0$ is irreducible under $R_0$, {\it any\/} $R_0$-invariant bilinear form has to be a multiple of this, and so, $B_0$ being skew symmetric, we may
conclude that $B_0=0$. Now $V_1={\q}_p$ and we take $\beta=1$ as the basis vector for $V_1$. Let $f(x)=B(x,\beta), x\in V_0$. Now $(v, W)$ acts on $V_0$ as $x\mapsto Wx$ and on $\beta \in V_1$ as $\beta \mapsto v+\beta$ and so the
condition for invariance is
$$
B(x, \beta)=B(Wx, v+\beta)
$$
for all $x, v\in V_0$. Thus, as we have already seen that $B_0=0$, we have
$$
f(x)=f(Wx)
$$
or that $f$ is an $R_0$-invariant linear form. By irreducibility of $V_0$ under $R_0$ we now have $f=0$. Since $B(\beta, \beta)=0$ as $B$ is skew symmetric, we have proved that $B=0$.

\vspace{14pt}

\noindent From \cite{var2} we know that $H^1(R, V')$  is a vector space over ${\q}_p$ and is isomorphic to ${\q}_p$:
$$
H^1(R, V')\simeq {\q}_p.
$$
In \cite{var2} the cocycles that describe this one-dimensional cohomology were explicitly given. For $\tau\in {\q}_p$ let
$$
\theta_\tau (v, W)=(2\tau v, -\tau (v,v))
$$
where the right side is interpreted as an element of $V'$ according to the conventions established earlier. It is then directly verifiable that the $\theta_\tau$ are in $Z^1(R, V')$, and the result of \cite{var2} is that
$$
\tau\longmapsto [\theta_\tau]
$$
is an isomorphism of ${\q}_p$ with $H^1(R, V')$. If $\psi$ is the additive character of ${\q}_p$ fixed earlier, then
$$
\psi \circ \theta_\tau
$$
is the corresponding element of $Z^1(R, V^\ast)$.

\vspace{14pt}

\noindent We can now determine the multiplier $\mu_\tau$ corresponding to the $\theta_\tau$ by the isomorphism of Theorem \ref{theorem theorem 1 moscow paper} Let
$$
r=((u, \eta), (v, W)), \quad r'=((u', \eta '), (v', W')).
$$
Then
$$
\mu_\tau (r,r')=\psi \bigg (\theta_\tau ((v, W)^{-1})(u', \eta ')\bigg ).
$$
But
$$
\theta_\tau ((v, W)^{-1})=\theta_\tau (-W^{-1}v, W^{-1})
=(-2\tau W^{-1}v, -\tau (v, v))
$$
so that
$$
\theta_\tau ((v, W)^{-1})(u', \eta ')=-2\tau (W^{-1}v, u')-\tau \eta '(v, v)=-2\tau (v, Wu')-\tau \eta '(v, v).
$$
Hence
$$
\mu_\tau (r,r')=\psi \bigg (-2\tau (v, Wu')-\tau \eta '(v, v)\bigg ).
$$
\par
In view of Theorem \ref{theorem theorem 1 moscow paper} we have the isomorphism
$$
H^2(G)\approx H^2(R)\times H^1(R, V').
$$
Now $R$ itself is a semidirect product $V_0\rtimes R_0$ but now $R_0$ is semisimple. As $R_0$ acts irreducibly on $V_0$ with a {\it symmetric\/} non-degenerate invariant bilinear form, we see as before that $0$ is the only
invariant {\it skew symmetric\/} invariant bilinear form. Hence all multipliers of $R$ are trivial when restricted to $V$. Thus by Theorem \ref{theorem theorem 1 moscow paper} we have
$$
H^2(R)\approx H^2(R_0)\times H^1(R_0, V_0').
$$
But $R_0$ is connected semisimple and so, by Theorem 3 of Section 6 of \cite{var2} we have
$$
H^1(R_0, V_0')=0.
$$
Hence
$$
H^2(R)\approx H^2(R_0).
$$
In other words, every multiplier of $R$ is equivalent to a lift to $R$ of a multiplier of $R_0$.

\vspace{14pt}

\noindent These remarks allow us to give a complete explicit description of $H^2(G)$. Let $n_0$ be a multiplier for $R_0$. We lift $n_0$ to the multiplier $n$ of $G$ by the composition of the maps
$$
G\longrightarrow R,\qquad R\longrightarrow R_0.
$$
We then define the multiplier
$$
m_{n_0, \tau}=n\mu_\tau
$$
of $G$. Thus
$$
m_{n_0, \tau}(r, r')=n_0(W, W')\psi \bigg (-2\tau (v, Wu')-\tau \eta '(v, v)\bigg )
$$

\vspace{14pt}

\noindent We now describe the Galilean particles.  First we fix $\tau \neq 0$.  The affine action corresponding to the cocyle $\theta_\tau$ is given by
$$(v,W): (\xi, t) \mapsto (W\xi + 2 \tau v, t - (W \xi,v) - \tau(v,v)).$$
The function $M: (\xi, t) \mapsto (\xi, \xi) + 4\tau t$ maps $V$ to $k$ and is easily verified to be invariant under the affine action. Hence the level sets of $M$ are invariant under the affine action. Since $M((0,a \slash 4 \tau
)) = a$, we see that $M$ maps onto $k$.

\vspace{14pt}

\noindent Fix $a\in {\q}_p$ and consider the level set
$$
M[a]=\big \{(\xi, t)\ \big|\ M(\xi, t)=a\big \}.
$$
The element $(0, a/4\tau)\in M[a]$; if $(\xi, t)\in M[a]$ then the element
$$
(\xi/2\tau, I)
$$
of $R$ sends $(0, a/4\tau)$ to $(\xi, t)$ by the affine action, as is easily verified.  Hence $M[a]$ is a single orbit. The orbits are thus all closed and so, by Effros's theorem the affine action is regular. One can see this also
explicitly by observing that the set
$$
\big \{(0, b)\ \big |\ \big (b\in {\q}_p)\}
$$
meets each affine orbit in exactly one point. The stabilizer in $R$ of $(0, a/4\tau)$ is $R_0$.
\vspace{14pt}

\noindent

\vspace{14pt}

\noindent For a given orbit the corresponding $m_{n_0,\tau}-$representations are parameterized by the $n_0$-representations of $R_0$.  However, as we shall now show, these representations are {\it projectively the same for
different $a$.\/} To see this we observe first that the projection map
$$
(\xi, t)\longmapsto \xi
$$
is a bijection of the level set $M[a]$ onto $V_0'$; in fact, the point
$$
\bigg (\xi, {a-(\xi, \xi)\over 4\tau}\bigg )
$$
is the unique point of $M[a]$ above $\xi$. The affine action on $M[a]$ corresponds to the action
$$
(v, W), \xi\longmapsto W\xi+2\tau v.
$$
We shall therefore identify $M[a]$ with $V_0'$ and the affine action by the above action. Hence, Lebesgue measure $\lambda$ is invariant.  {\it We note that the parameter $a$ has disappeared in the action.\/} Hence, by Theorem
\ref{theorem Mackeys theorem for m-reps},  the action of $R$ in the representation corresponding to the cocycle $m_{n_0, \tau}$ takes place on $L^2(V_0', {\kk}, \lambda)$ where ${\kk}$ is the Hilbert space for the $n_0$-representation of $R_0$ and is independent of
$a$. Furthermore, by the same theorem, the translation action by $(u,\eta)$ is just multiplication by
$$
\psi ((u, \xi)+t\eta)
$$
on $M[a]$ which reduces to multiplication by
$$
\psi \bigg( (u,\xi) +{\eta(a-(\xi,\xi))\over 4\tau} \bigg ).$$
on $L^2(V_0', {\kk}, \lambda)$. We now notice that
the factor
$$
\psi \bigg( {\eta a\over 4\tau } \bigg )
$$
{\it is independent of the variable $\xi$\/} and so it is a phase factor. It can therefore be pulled out and the remaining part is independent of $a$. Hence, projectively the entire representation can be written in a form that is
independent of the parameter $a$.  This proves that the representations with different $a$ are projectively equivalent and describe the same particle.

\vspace{14pt}

\noindent The relevant parameters are thus $\tau (\neq 0)$ and the projective representations $\mu$ of $R_0$. We interpret $\tau$ as the \textit{Schr$\ddot{o}$dinger mass} and $\mu$ as the \textit{spin}.

\vspace{14pt}

\noindent We still have to consider the case $\tau=0$ when the multiplier is the lift to $G$ of a multiplier $n_0$ for $R_0$ via the maps
$$G\longrightarrow R, \qquad R\longrightarrow R_0.$$
The affine action is now the ordinary action
$$
(v, W), (\xi, t)\longmapsto (W\xi, t-(W\xi, v).
$$
The function
$$
N : (\xi, t)\longmapsto (\xi, \xi)
$$
is clearly invariant and maps onto ${\q}_p$. We claim that the level sets of $N[a]$  where $N$ takes the values $a$ are orbits. The subset when $t=0$ is clearly an orbit for $R_0$. If $(\xi, t)\in N[a]$, select $v\in V_0$ such
that $(\xi, v)=-t$; then the element $(v, I)\in R$ takes $(\xi, 0)$ to $(\xi, t)$.  There is obviously an invariant measure on $N[a]$, namely the measure
$$
d\sigma_a\times dt
$$
where $d\sigma_a$ is the \lq\lq surface\rq\rq\  measure on the subset in $V_0'$ where $(\xi, \xi)$ takes the value $a$ (the \lq\lq sphere\rq\rq\ ). The spectrum is thus contained in a subvariety of $V_0'$. Over ${\r}$ this leads
to unphysical relations between momenta \cite{var1}. Over ${\q}_p$ there is no such argument but the representations do not seem to represent particles.

\section{The conformal group and conformal space time}\label{section poincare group and the conformal group}
\subsection{Imbedding of the Poincar\'{e} group in the conformal
group}\label{section imbedding poincare to conformal}
\begin{theorem}\label{theorem imbedding}
Let $k$ be a field of ch $\neq 2$. Suppose $W$ and $V$ are two Witt equivalent quadratic vector spaces over $k$
with $dim(V)=dim(W)+2$ and let $p \in V$ be a null vector.  Denote
by $H_p$ the stabilizer of $p$ in $\SO(V)$.  Then there exists an
isomorphism of algebraic groups
$$h:P_W \; \widetilde{\rightarrow} \; H_p$$ over $k$.
\end{theorem}
\begin{proof}
\noindent Fix a null vector $q \in V$ such that $(p,q)$ is a hyperbolic pair
in $V$ and let $W_p=\langle p,q\rangle ^\perp$. Then $V = W_p \oplus \langle p,q\rangle $ and
$W_p \simeq W$.  For brevity we  write $W$ for $W_p$.

\vspace{14pt}

\noindent Let $h$ be in $H_p$. We want to write $h$ in an explicit block
matrix form with respect to $V = \langle p\rangle  \oplus \langle q\rangle   \oplus W$.  Let $R\in {\rm End}(W)$ be   defined by $ht\equiv Rt$ mod  $\langle p,q\rangle $ for $t \in W$. A
calculation shows $hp = p$, $hq = - \frac{(t,t)}{2} p + q + t$, $hw
= - (t,Rw)p + Rw$ for $w\in W$. Let $e(t,R) \in  \textrm{ Hom}(W,\langle p\rangle )$ be the map $e(t,R):w
\mapsto -(t,Rw)p$. Then one can write the matrix of $h$ as
$$h =h(t, R)= \left(
  \begin{array}{ccc}
    1 & - \frac{(t,t)}{2} & e(t,R) \\
    0 & 1 & 0 \\
    0 & t & R \\
  \end{array}
\right).$$    Since $1=\det (h)=\det (R)$ and $(w, w)=(hw, hw)=(Rw,Rw)$, it follows that  $R\in \SO(W)$.

\vspace{14pt}

\noindent We note that $h$ is completely determined by $t$ and $R$.  Moreover, for any $t \in W, \; R \in \SO(W)$,
$h=h(t,R)$ as defined above makes sense and has the following properties:
\par
(1) $hp=p$, $hq$ is a null vector, and $(hq, p)=1$.
\par
(2) $hw\perp p, hw\perp hq, (hw, hw)=(w, w)$.
\par\noindent
These properties are sufficient to ensure that $h$ preserves the form on $V$. From the formula for $h$ we see that $\det(h)=1$ and so $h\in {\rm SO}(V)$. Since $hp=p$ we see finally that $h\in H_p$.

\vspace{14pt}

\noindent It is now trivial to verify that h is a homomorphism from
$P_{W}$ to $H_p$, i.e., $$h(t,R) \cdot h(t',R') = h(t+Rt',RR').$$
We omit the calculation. Thus $h$ is a morphism of algebraic groups $P_W\longrightarrow H_p$ which is defined over the ground field $k$ and is bijective. The inverse map is a morphism of algebraic varieties because it can be seen
as the restriction to $H_p$ of the map from a closed subvariety of ${\rm GL}(V)$ to $W\rtimes {\rm GL}(W)$ defined by
$$
\left(
  \begin{array}{ccc}
    1 & b & c \\
    0 & 1 & 0 \\
    g & t & R \\
  \end{array}
\right)\longmapsto (t, R).
$$
 We thus
see that we have an isomorphism of algebraic groups from $P_{W}$ to $H_p$, defined over $k$.
\end{proof}
\subsection{Conformal compactification of space time}\label{section compactification of spacetime}
Let $W, V$ be as above. Let
$$
G={\rm SO}(V).
$$
We shall now construct a smooth irreducible projective variety $[\Omega]$ such that
\begin{enumerate}
\item There is a $k$-imbedding of $W$ as a Zariski open subspace $A_W$ of $[\Omega]$
\item The group $G$ acts transitively on $[\Omega]$ and there is a $k$-isomorphism of $P_W$ with a  subgroup $G_W$ of $G$ which leaves $A_W$ invariant
\item The action of $G_W$
on $A_W$ is isomorphic (via the imbedding) to the action of $P_W$ on $W$
\end{enumerate}
The metric of $W$ does
not extend to $[\Omega]$; rather at each point $[x]$ of $[\Omega]$
we have a family of metrics differing by scalar multiples that
contains the metric of $W$ on $A_W$. The group $G$ preserves this family of metrics. Thus we say that $[\Omega]$ has \emph{a
conformal structure}; and as $G$ keeps this structure invariant we
call $G$ \emph{the conformal group}.  We refer to $([\Omega], G)$ as the {\it conformal compactification\/} of $(W, P_W)$. When $k$ is a {\it local field\/}, $[\Omega]$ (or rather, the set of its $k$-points) is compact, thus
justifying our terminology. These ideas are
summarized in the following theorem.
\begin{theorem}\label{compactification}Given two Witt equivalent quadratic vector spaces $W$ and $V$ over
$k$ with $dim(V)=dim(W)+2$ there exists a conformal compactification of $(W, P_W)$.
\end{theorem}
\noindent We prove this theorem in series of lemmas.
%
%
\begin{definition} Let $V, W$ be as above.  We define
$$\Omega =M_0= \{p \in V \; \vert \; p \neq 0 ,
(p,p)=0 \}.$$\end{definition}
\noindent There is a basis of $V$ for which the quadratic form
becomes:
$$Q(x)=a_0 x_0^2 + a_1 x_1^2 + ... + a_{n+1} x_{n-1}^2, \;\; a_i \neq 0$$
where $n=\dim (V)$.
\noindent Thus the equation defining $\Omega$ is
$$a_0 x_0^2 + ... + a_{n-1} x_{n-1}^2 = 0.$$  This
homogeneous polynomial defines a smooth irreducible quadric cone $[\Omega]$  of dimension $n-2$ in the
projective space $\mathbb{P}(V)$. Let $P(x\longmapsto[x])$ be the
map from $V\setminus \{0\}$ to $\mathbb{P}(V)$. Then  $[ \Omega ] $ is the
image under $P$ of $\Omega$ in $\mathbb{P}(V)$, and  is stable under the
action of $\SO(V)$. The tangent space at $x \in \Omega$ is $V_x =
\{v \in V \; \vert \; (x,v)=0 \}$, and for $[x] \in [ \Omega ]$,
the
tangent space at $[x]$ is $[\Omega]_{[x]}$ and is defined as the
image of the tangent map $dP_x$ of $V_x$.
\begin{lemma}
$[\Omega]$ has a natural $G$-invariant conformal structure.
\end{lemma}
\begin{proof}
We note that tangent map $dP_x: V_x \rightarrow [\Omega]_{[x]}$ is
surjective because $P$ is submersive. Hence, the kernel of $dP_x$ is
one dimensional. We know that $P$ is constant on the line $kx$ so
$dP_x$ vanishes on $kx$.  Thus the kernel of $dP_x$ is the line
$kx$. Hence, the quadratic form $Q$ on $V$ induces a quadratic form
$\tilde{Q}$ on $[\Omega]_{x}$.  We note that if we use the map
$dP_{\lambda x}: V_{\lambda x} \rightarrow [\Omega]_{[x]}$ to define
the induced quadratic form $\tilde{Q}'$ then $\tilde{Q}' =
\lambda^2\tilde{Q}$. Furthermore, if we have $g \in \SO(V)$ and $x' =
\lambda x$, then the set of metrics at $[\Omega]_{[x]}$ induced from
$V_x$ goes over to the set of metrics induced from $V_{\lambda x}$.
Thus $[\Omega]$ has a conformal structure defined by these induced metrics. The definition of the conformal structure makes it clear that it is $G$-invariant.
\end{proof}

\vspace{14pt}

\noindent
We write
$$
V=W\oplus \langle p, q\rangle
$$
where the sum is orthogonal, and $\langle p, q\rangle$ is hyperbolic with
$$
(p,p)=(q,q)=0, \quad (p,q)=1.
$$
We define $A_{[p]} = \{[a] \in [\Omega] \; \vert \; (p,a) \neq 0 \}$ and
we introduce $C_p$ as the set of null vectors of $V_p$. Thus $C_p=V_p\cap \Omega$. Write $C_{[p]}$
for the image of $C_p$ in $[\Omega]$. Then we have
$A_{[p]} = [\Omega] \backslash C_{[p]}$ since $V_p$ is defined by the equation $(p, v)=0$. Let $[a] \in A_{[p]}$, we
write $a = \alpha p + \beta q + w$, where $w \in W$, then, as $(p,a)
\neq 0$, we must have $\beta \neq 0$. A quick calculation shows that
$\alpha = \frac{-(w,w)}{2}$.  Since we are only interested in the
image of $a$ in the  projective space we may take $\beta$ to be $1$.
Then $[a]$ is given by $[\frac{-(w,w)}{2} : 1 : w]$ so $[a]$ is
entirely determined by $w$. We thus have the bijection
$$
J : W  \simeq
A_{[p]}\qquad J: w \mapsto \bigg [{-(w, w)\over 2}p+q+w\bigg ].
$$
\begin{lemma}$A_{[p]}$ is a Zariski open dense subset of $[\Omega]$.\end{lemma}
\begin{proof}
 It is clear that $A_{[p]}$ is a Zariski open subset of $[\Omega]$;
 it is dense since $[\Omega]$ is irreducible.
\end{proof}
\begin{lemma}Let $H_p$ be the subgroup of ${\rm SO}(V)$ that fixes $p$. Then $H_p$ leaves invariant the image $A_{[p]}$ of $W$ under $J$.  Moreover the map $J$ intertwines the actions of $(t,R)$ on $W$ and $h(t,R)$ on $A_{[p]}$ (see theorem \ref{theorem imbedding}).\end{lemma}
\begin{proof}
Notice first that if $h\in H_p$ then $h$ stabilizes $V_p$, therefore $C_p$, hence $A_{[p]}$. Let $(t,R)$ be in $P_{W}$.
\end{proof}
\noindent All
claims of Theorem \ref{compactification} have now been proven.
\begin{lemma}\label{lemma padic ball is compact}  If $k=\Q_p$ then $[\Omega]$ is
compact.\end{lemma}
\begin{proof}
Since $[\Omega]$ is a closed subset of $\mathbb{P}(\Q_p^{n+1})$ and
$\mathbb{P}(\Q_p^{n+1})$ is compact, $[\Omega]$ is compact.
\end{proof}
\noindent Lemma \ref{lemma padic ball is compact} shows that if the underlying
field is $\Q_p$ then the projective imbedding becomes the
compactification of spacetime.
\subsection{Conjugacy of imbeddings}\label{section Conjugacy of
imbeddings}
The following theorem is a converse of sorts to Theorem
\ref{theorem imbedding}. It states that the subgroups of $\SO(V)$ that are
isomorphic to a Poincar\'{e} group $P_W$ arise only as stabilizers
of null vectors.  $\SO(V)$ acts by conjugacy transitively on the set
of all Poincar\'{e} groups inside $\SO(V)$.

\begin{theorem}\label{theorem conjugacy}
Let $W$ and $V$ be two quadratic vector spaces with $W$ Witt
equivalent to $V$ with $dim(V)=dim(W)+2$. If there is
an imbedding $f:P_W \hookrightarrow \SO(V)$ of algebraic groups over
$k$, then for $\dim (W)\geq 3$,
\begin{enumerate}
\item $f(P_W) = H_p$, where $H_p$ is a stabilizer of some null vector
$p \in \SO(V)$.
\item All such imbeddings $f$ are conjugate under
$\SO(V)(k)$.
\end{enumerate}
\end{theorem}

\noindent Theorem \ref{theorem conjugacy} is of great theoretical interest. But its proof is long and technically involved.
Since neither the theorem nor its proof itself are used in the rest of the paper, we postpone the proof
to the Appendix.

\subsection{Partial conformal group}\label{section PCG} In this
section we introduce a subgroup $\tilde{P}_V$ of $\SO(V)$, the
stabilizer of the line $kp$.  It is an easy computation that $h \in
\SO(V)$ stabilizes the line $kp$ iff $h$ has the form
$$h = \left(
  \begin{array}{ccc}
    c & -c \frac{(t,t)}{2} & ce(t,R) \\
    0 & \frac{1}{c} & 0 \\
    0 & t & R \\
  \end{array}
\right)$$ Here $c \in k^{\times}$, $t \in W$, $R \in \SO(W)$ and $e(t,R) \in
\textrm{ Hom}(W,\langle p\rangle )$.  We write $h=h(c,t,R)$. We denote the set of
all such matrices $h$ as
$$\tilde{P}_V =\{ h(c,t,R) \; \vert \; c \in k^{\times}, \; t\in W, \; R\in \SO(W) \}.$$
Let us denote by $\tilde{c}$ the matrix
$$\tilde{c} = \left(
  \begin{array}{ccc}
    c & 0 & 0 \\
    0 & \frac{1}{c} & 0 \\
    0 & 0 & I \\
  \end{array}
\right) \;\;\; (c \in k^{\times}).$$
\noindent Given $h(c,t,R) \in \tilde{P}_V$ then $h(c,t,R) = \tilde{c} h(t,R)$,
where $h(t,R) \in P_V$.

\vspace{14pt}

\noindent The following is immediate.
\begin{lemma}\label{commutation rules} $\tilde{P}_V = \{ \tilde{c} h(t,R) \; \vert \; c \in k^{\times} \; , \;
t \in W \; , \; R \in \SO(W) \}$.
\begin{enumerate}
\item $\tilde{P}_V \simeq
V\rtimes (\SO(V) \times k^\times)$.
\item Multiplication is given by: $\tilde{c}h(t,R)\tilde{c'}h(t',R') = \widetilde{cc'}h(\frac{1}{c'}t + R t', RR')$
\item The conjugation action of $\tilde{c}$ on the translation part is to
dilate it by a factor of $c$.  That is
$\tilde{c}h(t,R)\widetilde{c^{-1}} = h(ct,R)$.  Note that
$\tilde{c}$ commutes with the $R$ action.
\end{enumerate}
\end{lemma}
\begin{lemma} $\tilde{P}_V$ is the largest subgroup of $\SO(V)$ that leaves
$A_{[p]}$ invariant.\end{lemma}
\begin{proof}
We note that it is easier to work with $A_p:=\{ a \in \Omega \; \vert \; (p,a) \neq 0 \}$.  Let $g$ be any element of $\SO(V)$ that leaves $A_p$  invariant.  Then $g$ leaves $A_{[p]}$ invariant as well.
We want to first show that $g p = \alpha p +w$ where $w\in W$, which
is equivalent to showing that $gp \in V_p$. If $g$ preserves
$A_p$, then $g$ also preserves the compliment of $A_p$,
which is $V_p \cap C_p$. Now $p \in V_p \cap C_p$ so that $gp \in V_p
\cap C_p$.

\vspace{14pt}

\noindent If $g$ preserves $A_p$, it also preserves $A_p \backslash C_p = V_p \cap C_p$.
We must show that  $g\cdot\langle p \rangle = \langle p \rangle$.  If $\langle p \rangle$ is the only null line in $V_p$, then $g\cdot \langle p \rangle = \langle p \rangle$ trivially.  So assume that $V_p = \langle p \rangle + W$
has other null lines. Now $W \cap C_p$ is stable under $\SO(W)$ and $\SO(W)$ acts irreducibly on $W$, so $W \cap C_p$ spans $W$. We have that $g(W) \subset$ Span$(g(W \cap C_p)) \subset$ Span$(g(V_p \cap C_p)) \subset$ Span$(V_p \cap C
\subset V_p)$.  Hence, $g(W) \subset V_p$.  On the other hand, as $p \in V_p \cap C_p$, $g \cdot p \in V_p$. So $g(V_p) \subset V_p$ and $g(V_p^\perp) \subset V_p^\perp$.  Hence, $g\langle p \rangle = \langle p \rangle$.
\end{proof}
\begin{definition}We will call $\tilde{P}_V$ the partial conformal
group.\end{definition}
\noindent This is a reasonable definition since $\tilde{P}_V$ is the subgroup
stabilizing $A_{[p]}$.
\section{Extendability of PUIRs of the Poincar\'{e} group to the PUIRs of the conformal group}\label{section extendability of PUIRs of poincare group to conformal}
As we discussed in Section \ref{section imbedding poincare to conformal}, if $V_1$ and $V_0$ are two quadratic vector
spaces with $V_1$ Witt equivalent to $V_0$ with $dim(V_0) =
dim(V_1)+2$, then the Poincar\'{e} group $P_{V_1}$, can be imbedded
as a subgroup of the conformal group $\SO(V_0)$, and furthermore,
that any two such imbeddings are conjugate over $\SO(V_0)$.  A
natural question that one may ask is the following: are there PUIRs
of the Poincar\'{e} group that extend to be PUIRs of the conformal
group? PUIRs that do extend to the conformal group are said to have
conformal symmetry. Classically, only massless particles (photons)
have conformal symmetry and the corresponding PUIRs of the real
Poincar\'{e} group extend to PUIRs of the real conformal group
\cite{ang}. We would like to explore this question in the
$p$-adic setting. Our ultimate goal is to establish some necessary
conditions for this extension to be possible.
\begin{definition} Let $V_1$ and $V_0$ be two Witt
equivalent quadratic vector spaces over $\Q_p$ with $dim(V_0) =
dim(V_1)+2$. When a PUIR $U$ of $P_{V_1}$ can be extended to be a
PUIR $V$ of $\SO(V_0)$ we say that the particle corresponding to $U$
has \textit{conformal symmetry}.
\end{definition}
\begin{definition}
When a PUIR of $U$ of $P_V$ can be extended to be a PUIR $\tilde{U}$ of the
group $\tilde{P}_V$, we say that the particle corresponding to $U$
has \textit{partial conformal symmetry}. \end{definition}
\noindent We make the following trivial, but important observation: If a
particle does not have partial conformal symmetry, then it does not have
conformal symmetry. In the next section we aim to establish some
necessary conditions for a particle to have partial conformal
symmetry.
\subsection{Extensions of $m$-representations of semidirect
products}\label{section Extensions of $m$-representations of semidirect
products} Let $A,L$ and $M$ be lcsc groups with $A$ being abelian
and $L$ being a closed subgroup of $M$. Suppose $M$ acts on $A$ so
that we may form the semidirect products $G = A \rtimes L$, $H = A
\rtimes M$. We assume: $a)$ all multipliers of $G$ and $H$ are
trivial on $A$; $b)$ that $H^1(L,A^*)=0$, $H^1(M,A^*) = 0$; and $c)$
$1$ is the only character of $A$ fixed by $L$; and $d)$ The actions
of $M$ and $L$ on $A^*$ are regular.

\vspace{14pt}

\noindent Because of the assumptions that $H^1(G,A^*) = 0$, and $H^1(H,A^*)=0$, and
that the actions of $M$ and $L$ are regular, irreducible
$m$-representations $U$ of $G$ (resp.\,$H$) correspond to pairs
$(\chi,u)$ where $\chi \in A^*$ and $u$ is an irreducible
$m_\chi$-representation of the stabilizer $G_\chi$ (resp.\,$H_\chi$)
of $\chi$ in $G$ (resp.\,$H$).

\vspace{14pt}

\noindent We will need the following technical result.
\begin{lemma}Let $U$ be an $m$-representation of $G$ where $m$
is standard.  Let $V_1$ be an $m_1$-representation of $H$ extending
$U$.  Then we can find a standard multiplier $m'$ for $H$ such that
$m'\vert_{G \times G} = m$ and $U$ has an extension $V$ to $H$ as an
$m'$-representation with $V(ah) = F(ah)V_1(ah) \; (ah \in H)$ for
some Borel function $F: H \rightarrow T$ with $F=1$ on $G$.
\end{lemma}
\begin{proof}
From Mackey's work (see \cite{mac2}) we know that $V:ah \mapsto
m_1(a,h)V_1(ah)$ is an $m'$-representation of $H$ with
$m'\vert_{A\times A} =1$ and $m'\vert_{A\times H} =1$. Clearly, $V$
extends $U$, $m' \simeq m_1$ and $m'\vert_{G \times G}=m$.  As
$H^1(M,A^*)=0$, we have $m'(ah,a'h')=m_0'(h,h')f(h[a']) \slash
f(a')$ where $m_0'$ is a multiplier for $M$ and $f\in A^*$. Since
$m'\vert_{G \times G}=m$, $f(g[a'])f(a')^{-1}=1 \;\; \forall g\in L,
\; a'\in A$.  Hence, $f =1$ by the assumption that $1$ is the only
character fixed by $L$. Thus $m'$ is already standard.
\end{proof}
\noindent The following is a key lemma that will be utilized often to prove
the impossibility of the extension of both massive and eventually
massive particles.
\begin{lemma}\label{lemma extension} Let $U$ be an irreducible $m$-representation of $G$
for a standard multiplier $m$ for $G$.  Let $U$ correspond to the
$L$-orbit of $\chi \in A^*$ and an irreducible
$m_\chi$-representation $u$ of the stabilizer $L_\chi$ of $\chi$ in
$L$, $m_\chi$ being $m\vert_{L_\chi \times L_\chi}$.  Then the
following are equivalent:
\vspace{10pt}
\\
\noindent \textbf{(1)} $U$ extends to a projective unitary representation $V_1$ of $H$.\\
\\
\noindent \textbf{(2)}\textbf{(a)} $M [ \chi ] \setminus L [ \chi ]$ is a null set in $M[ \chi ]$.\\
\\
\noindent\hspace*{14pt}\textbf{(b)}  There is a standard multiplier $m'$ for $H$ with $m'\vert_{G\times G} = m$.\\
\\
\noindent\hspace*{14pt}\textbf{(c)} $u$ extends to a $m'_\chi$-representation of $M_\chi$.\\
\\
\noindent In this case there is an $m'$-representation $V$ of $H$
such that $V$ belongs to the same equivalence class as $V_1$ with:
\begin{description}
\item[(I)] $V \vert_G =U$.
\item[(II)] $V$ corresponds to $\chi$ and $v$ where $v$ is an
$m'_\chi$-representation of $M_\chi$.
\item[(III)] $v \vert_{L_\chi} = u$.
\end{description}
\end{lemma}
\begin{proof}
$(1) \Rightarrow (2)$:  We may assume $U$ extends to an
$m'$-representation $V$ of $H$ belonging to the same equivalence
class as $V_1$ where $m'$ is standard and $m'\vert_{G \times G} =
m$. Clearly $V$ is irreducible.  Hence, the spectrum of $V$ lives on
an $M$-orbit in $A^*$.  But as $V$ and $U$ have the same restriction
to $A$, the spectrum of $V$ must meet $L [ \chi ]$ so that we may
assume it to be $M [ \chi ]$.  But then $M [\chi] \setminus L
[\chi]$ must be null.  This proves $(2)(a)$ and (2)(c).

\vspace{14pt}

\noindent We may now write $V$ in the form:
$$(V(ah)f)(\zeta) =  \langle a,
\zeta\rangle \rho_h(h^{-1}\zeta)^{\frac{1}{2}} C(h,h^{-1}\zeta)f(h^{-1}\zeta) , \;\;\;\; (\zeta \in M [\chi],
\; h \in M)$$ where $C$ is a strict $m'$-cocyle that defines the
$m'_\chi$-representation $v$.  On the other hand, $U$ is given by:
$$(U(ag)f)(\zeta) = \langle a,\zeta\rangle \rho_g(g^{-1}\zeta)^{\frac{1}{2}}
D(g,g^{-1}\zeta)f(h^{-1}\zeta)\;\;\;\;(\zeta \in L [\chi],\; g \in
L)$$ where $D$ is a strict $m$-cocyle defining the
$m$-representation $u$. Since $V \vert_G=U$, it follows that $D(g,
\nu) = C(g, \nu)$ for each $g$ for almost all $\nu \in M [\chi]$.
Hence, by Lemma \ref{lemma strict cocycle lemma}, $u$ is equivalent to the restriction of $v$ to $L_\chi$. If $u(g) = r v(g)r^{-1} \;\;
(g \in L_\chi)$, where $r$ is a unitary representation in the space
of $v$, it is clear that $u$ extends to $rvr^{-1}$.  This proves
$(2)(b)$.

\vspace{14pt}

\noindent $(2) \Rightarrow (1)$:  Extend $u$ to an $m'_\chi$-representation of
$M_\chi$ and build a strict $(M,M [ \chi])$-cocyle $C$ for the
multiplier $m'$ for $M$ that defines the $m'_\chi$-representation at
$\chi$.  The restriction of $C$ to $L$ is a strict cocycle for $m_1
= m' \vert_{L \times L}$.  The $m'$-representation of $H$
corresponding to $(\chi,m')$ restricts on $G$ to the
$m_1$-representation defined by $(\chi,m_1)$, and hence is
equivalent to $U$. So $U$ extends to a PUR of $H$.

\vspace{14pt}

\noindent  The above proof also establishes (I),(II) and (III).
\end{proof}
\subsection{Impossibility of partial conformal symmetry for massive
particles}\label{section Impossibility of partial conformal symmetry
for massive particles} We now show that massive particles do not
posses partial conformal symmetry.  We begin with some important
lemmas.
\begin{lemma}\label{lemma ptildeorbitmeasure}
The orbit of a massive point under $\SO(V) \times \Q_p^\times$ is
open in $V$.
\end{lemma}
\begin{proof}
Let $x \in V$ be such that $Q(x)=a \neq 0$; then if $g \in \SO(V) \times \Q_p^\times$ and $g [ x ] = tx$ then $Q(g [ x ]) = at^2$. Thus the
orbit of $Q(x)$ under $\tilde{P}$ is $a (\Q_p^\times)^2$. Hence, the
orbit of $x$ is $Q^{-1}(a (\Q_p^\times)^2)$. Since $Q$ is a
continuous function, the orbit of $x$ will be open in $V$ if we can
show that $a(\Q_p^\times)^2$ is open in $\Q_p^\times$. We note that
it suffices to prove that $(\Q_p^\times)^2$ is open in $\Q_p$. This
is an easy verification and we omit it here.
\end{proof}
\begin{lemma}\label{quasiinvarianthaar}
Let $p \in V$ be a massive point, then the quasi-invariant measure
class on the orbit $\SO(V) \times \Q_p^\times \cdot p$, is the
Lebesgue (Haar) measure class.
\end{lemma}
\begin{proof}
By Lemma \ref{lemma ptildeorbitmeasure} the orbit $\SO(V) \times
\Q_p^\times [ p ] = \omega_p$ is open in $V$.  Let $E \subset
\omega_p$ be a set of Haar measure 0 in $V$. Since $\omega_p$ is
open in $V$, the Haar measure is also defined on $\omega_p$. Let
$\mu$ be the Haar measure.  Since $\SO(V) \times \Q_p^\times$ acts
linearily for any $(g,c) \in \SO(V) \times \Q_p^\times$, we have that
$\mu((g,c) \cdot E) = \vert \text{ det }((g,c)) \vert_p \mu(E)$. Hence,
$\mu((g,c) \cdot E)=0$.  Thus the measure class on $\omega_p$ is the
Haar measure class, and it is quasi-invariant under $\SO(V) \times
\Q_p^\times$.
\end{proof}
\begin{corollary}\label{corollary orbitsof0measure}
Both massive and massless orbits under $\SO(V)$ have Haar measure 0
in $V$.
\end{corollary}
\begin{proof}
Let us now take $x \in V$ such that $x \neq 0$ and $Q(x) = a$. Then
$f(x) = Q(x) - a$ is an analytic function and defines a subset $Q_a
= \{ x \in V \; \vert \; f(x) =0 \}$ of $V$.  We want to show that $Q_a$
has measure $0$ in $V$.  We may assume that there is a basis of $V$,
$(e_i)$, such that $(e_i,e_j) = a_i \delta_{ij}$ and that $Q(x) =
\sum a_i x_i^2$. Since $x = 0$ is not in $Q_a$ we see that $Q$ has a
nonzero gradient on all of $Q_a$.  It follows that $Q_a$ has Haar
measure 0 in $V$.
\end{proof}
\begin{theorem}\label{massiveparticlesdonotextend}
A massive PUIR of $P_V$ does not have extension to $\tilde{P}_V$
\end{theorem}
\begin{proof}
Let $p$ correspond to a massive orbit.  We have that $P_V=V \rtimes
\SO(V)$ and $\tilde{P}_V = V \rtimes (\SO(V) \times \Q_p^\times)$. Let
us denote $\SO(V)$ by $L$ and $\SO(V) \times \Q_p^\times$ by $M$. If a
massive PUIR of $P_V$ were to have an extension to $\tilde{P}_V$
then by Lemma \ref{lemma extension} $M [ p ] \setminus L [ p ]$
would be a null set. However, by Corollary \ref{corollary
orbitsof0measure} $L [ p ]$ has Lebesgue (Haar) measure $0$, so $M [
p ]$ would have to have Lebesgue measure $0$.  But by Lemma
\ref{lemma ptildeorbitmeasure} the orbit $M [ p ]$ is open and so
has nonzero Lebesgue measure. Hence, the massive representations cannot extend even to the partial conformal group and therefore cannot
extend to the full conformal group.
\end{proof}
\subsection{Impossibility of partial conformal symmetry for
eventually massive particles}\label{section Impossibility of partial
conformal symmetry for eventually massive particles} Since the
massive particles do not have partial conformal symmetry, we now
turn our attention to massless particles.

\vspace{14pt}

\noindent Let $V$ be a quadratic vector space and $p$ a nontrivial null vector
in $V$.  Let $P_V=V \rtimes \SO(V)$ and let $\tilde{P}_V = V \rtimes
(\SO(V) \times k^\times)$. As discussed in Section \ref{section PCG}
the action of $c \in k^\times$ on $v \in V$ is $c: v \mapsto cv$,
and $k^\times$ commutes with $\SO(V)$. We know from Theorem
\ref{theorem imbedding} that the stabilizer of $p$ in $\SO(V)$ is isomorphic
to $P_{V_1} = V_1 \rtimes \SO(V_1)$ where $V_1$ is a vector space
Witt equivalent to $V$ with $dim(V)=dim(V_1)+2$. We now claim the
following:
\begin{proposition}
Let $\tilde{P}_{Vp}$ be the stabilizer of $p$ in $\SO(V) \times
k^\times $. Then $\tilde{P}_{Vp}$ is isomorphic to
$\tilde{P}_{V_1}$.
\end{proposition}
\begin{proof}
Let $(g,c) \in \SO(V) \times k^\times$. Then $(g,c)$ acts on $p$ by
$(g,c): p \mapsto cg[p]$. Hence, $(g,c)$ fixes $p$ iff $g[p] =
\frac{1}{c}p$.  In other words, $g$ stabilizes the line $kp$. As
proven in Section \ref{section PCG} the stabilizer of the line is
$\tilde{P}_{V_1}$.
\end{proof}
\begin{lemma}We have $H^1(\SO(V) \times \Q_p^\times,V')=0$. Moreover, the action of $(g,c) \in \SO(V) \times
\Q_p^\times$ on $\lambda \in V'$ is by $(g,c): \lambda \mapsto cg
\cdot \lambda$.\end{lemma}
\begin{proof}
Let $F \in H^1(\SO(V) \times \Q_p^\times, V')$.  We set
$L=\SO(V)$ and write elements of $\SO(V) \times \Q_p^\times$ as
$(l,c)$.  Then $F((l,1))$ is a trivial cocycle for all $l \in L$,
since $H^1(\SO(V),V' ) =0$.  Thus we can find a $\lambda \in V'$ such
that $F((l,1)) = l [ \lambda ] - \lambda \;\; \forall \;l \in
\SO(V)$. If $F_1((l,c)) = F((l,c)) - ((l,c))[\lambda] - \lambda$,
then $F_1 \simeq F$ and $F_1$ is $0$ on $L$.  So we may assume that
$F$ is $0$ on $L$ to begin with. We may identify $(l,1)$ with $L$
and $(1,c)$ with $c \in \Q_p^\times$ and we can then write $(l,c) =
lc$. We now use the fact that $lc = cl$ to write:$$F((l,c)) = F(lc)
= F(l) + l [ F(c) ] = F(c) + c [ F(l)].$$ Since $F(l) =0$ we get
$F(c) = l [ F(c)]$.  However, $L$ does not fix any nontrivial vector
in $V'$ so we must have $F(c)=0$ and so $F=0$. \end{proof}
\begin{lemma}\label{lemma PUIRs come from massless PUIRs}
Suppose $p$ is a null vector in $V^*$ and $U$ is an irreducible
$m$-representation of $P_V$ corresponding to $p$ and an irreducible
$m_0$-representation $u$ of $\SO(V)_p = L_p$.  Suppose that $U$ has an
extension to $\tilde{P}$ as a projective unitary representation.
Then, identifying $\SO(V)_p$ with $P_{V_1} = V_1 \rtimes \SO(V_1)$,
$u$ is a massless PUIR of $P_{V_1}$ that has partial conformal
symmetry.\end{lemma}
\begin{proof}
By Lemma \ref{lemma extension} $2) \; c)$, $u$ extends to be a
representation of $\tilde{P}_{V_1} = V_1 \rtimes ( \SO(V_1) \times
\Q_p )$.  Now by Theorem \ref{massiveparticlesdonotextend} $u$ must
be massless.
\end{proof}
\noindent Lemma \ref{lemma PUIRs come from massless PUIRs} shows that if one
has a PUIR $U_0$ of a Poincar\'{e} group $P_{V_0}$ that extends to a
PUIR of the conformal group $\tilde{P}_{V_0}$, then $U_0$
corresponds to a PUIR $U_1$ of a stabilizer $P_{V_1}$ of a massless
$p_0 \in V_0^*$. We note that $P_{V_1} = V_1 \rtimes \SO(V_1)$ where
$V_1$ is a quadratic vector space Witt equivalent to $V_0$ with
$dim(V_0)=dim(V_1)+2$.  So $P_{V_1}$ is itself a Poincar\'{e} group.
Now in turn, $U_1$ has partial conformal symmetry and will
correspond to a PUIR $U_2$ of the stablizer $P_{V_2}$ of some $p_1
\in V_1$.  So it is clear that this process can be repeated until
one reaches a stage $R$ where $V_R$ is anisotropic. At the
anisotropic stage, the only massless character in $V_R^*$ is the
trivial one. One may also end the process by picking the trivial
null vector at any stage.  We thus have a chain of Poincar\'{e}
groups $P_{V_0},P_{V_1}, ...$ and corresponding massless
representations $U_0,U_1,...$.  From our discussion we have the
following theorem:
\begin{theorem}\label{theorem all sub representations are
massless}If $U$ is massless and has partial conformal symmetry, all
the $U_\nu$ are massless.\end{theorem}
\noindent We say that $U$ is \textit{eventually massive} if some $U_\nu$ is
massive.
\begin{theorem}\label{theorem eventually massive do not extend}Eventually massive particles do not have partial conformal
symmetry.\end{theorem}
\noindent Both theorems \ref{theorem all sub representations are massless} and
\ref{theorem eventually massive do not extend} are immediate from
Theorem \ref{massiveparticlesdonotextend} and Lemma \ref{lemma PUIRs
come from massless PUIRs}.

\appendix
\section{Proof of Theorem \ref{theorem conjugacy}}

\vspace{14pt}

\noindent
Write
$$
D=\dim (V), \quad d=\dim (W)=D-2,\quad G={\rm SO}(V),\quad L={\rm SO}(W).
$$
We want to show that $P_W$ fixes a non-zdero null vector in $V$. For then we will have a $k$-imbedding of $P_W$ with the stabilizer of this null vector, which must be an isomorphism as the groups have the same dimension by Theorem \ref{theorem imbedding} Moreover the conjugacy of the imbeddings will follow from the fact that the null vectors form a single orbit. It is even enough to show that $P_W$ fixes a nonzero vector over the algebraic closure $\bar k$ of the ground field $k$. Indeed, if this were assumed, then, $P_W$ fixes a nonzero vector defined over $k$ itself, because the action of $P_W$ in $V$ is defined over $k$. Let $v$ be such a vector and assume that it is not a null vector. We have seen in Chapter 4 (
that the stabilizer of $v$ is ${\rm SO}(U_v)$ where $U_v=(kv)^\perp$ is a non-degenerate quadratic vector space. Hence
$$
\dim (P_W)=\dim ({\rm SO}(U_v))={(D-1)(D-2)\over 2}
$$
which means that the imbedding $P_W\hookrightarrow {\rm SO}(U_v)$ must be an isomorphism. But this is impossible as ${\rm SO}(U_v)$ is semisimple while $P_W$ has a non-trivial radical. Hence $v$ must be a null vector.

\vspace{14pt}

\noindent We may therefore work over $\bar k$ for the rest of the proof. In other words {\it we shall assume that $k$ itself is algebraically closed from now on.\/}

\vspace{14pt}

\noindent From standard algebraic group theory one knows that the action of the additive group of $k$ are unipotent and so they fix some non-zero vectors. By induction on $r$ this is true for actions of $k^r$ and hence in every
$W$-stable subspace of $V$ we can find non-zero vectors fixed by $W$. Let $U$ be a $P_W$-stable subspace of $V$ of minimal dimension $r$. Clearly
$P_W$ acts irreducibly on $U$.  Let $U_1$ be the subspace of $U$ on which
$W$ acts trivially.  Clearly $U_1$ is $L$-stable, hence $P_W$-stable, so
$U_1=U$ by the minimality of $U$.   Thus $W$ acts trivially on $U$.
Since $U$ is $P_W$-stable so is  $U^\perp$. Hence, by minimality of $\dim (U)$, we must have
$$
D-r=dim(U^\perp)=r \geq
dim(U)
$$
so that $r \leq \frac{D}{2}$.

\vspace{14pt}

\noindent We can already complete the proof in characteristic $0$ except when $D=6$, since we have the following lemma. As we will consider the exceptional cases $5\le D\le 8$ below this exception is included in the argument where
the characteristic need not be $0$.
\begin{lemma} In characteristic $0$ every non-trivial irreducible representation of ${\rm SO}(n)$ for $n\ge 3, n\not=4$ has dimension $\ge n$.
\end{lemma}
\begin{proof} We may work over ${\c}$. It is enough to verify this for the fundamental representations. These are the representations induced on the exterior tensors plus the spin representations. We can exclude the spin
representations because they are not representations of the orthogonal group. The exterior representations have dimensions $n, {n\choose 2}$, etc all of which are $\ge n$. The condition $n\not=4$ is due to the fact that ${\rm
SO}(4)$ is not simple.
\end{proof}
\noindent Assuming the above lemma we observe that if $U$ carries a non-trivial representation of $L$, and $D-2\not=4$, then $r\ge D-2$, which, combined with $r\le D/2$, yields $D\le 4$. As $D\ge 5$ we must have that $L$ acts
trivially on $U$ and so every vector of $U$ is fixed by $P_W$. This finishes the argument.

\vspace{14pt}

\noindent
We now resume the proof in the case of arbitrary characteristic $\not=2$. We consider the restriction of the quadratic form to $U$.  Let $U_1$ be the radical of $U$.
$U_1$ is stable under $P_W$ and so by irreducibility of $U$, either
$U_1 = 0$ or $U_1 = U$. We claim that $U_1 = U$, i.e., $U$ is isotropic. Assume on the contrary that $U_1=0$.  In this case the quadratic form is
nondegenerate so $V = U \oplus U^\perp$.  Both $U$ and $U^\perp$ are
$P_W$-stable and thus $P_W \subset \ORTH(U) \times \ORTH(U^\perp)$.  But
since $P_W$ is a connected group it must be mapped to the connected
component of $\ORTH(U) \times \ORTH(U)^\perp$, so really we have that $P_W
\subset \SO(U) \times \SO(U^\perp)$. Hence, $dim(P_W) \leq dim(\SO(U)) +
dim (\SO(U^\perp))$. We now have
$$\frac{(D-1)(D-2)}{2} \leq \frac {r(r-1)}{2} +
\frac{(D-r)(D-r-1)}{2}$$ This gives, after a simple calculation,
$$D(r-1) \leq r^2-1 = (r+1)(r-1).$$
So either $r=1$ or $D \leq r+1$.  Note that we already know that $r
\leq \frac{D}{2}$.  So if $D \leq r+1$, then we must have that $D
\leq 2$.  Hence, $r=1$.  So $L$ must
be trivial on $U$, since $L$ is semisimple and there is no
nontrivial homomorphism into any abelian group.  Thus $P_W$ must fix
a basis vector $u$ of $U$ which is not null. We have already excluded this alternative. Thus $u$ is a null vector and $P_W$ imbeds into the normalizer in $\SO(V)$ of $U$. Hence $U$ is isotropic.

\vspace{14pt}

\noindent Let $r =
dim(U) \geq 1$.  We now investigate the structure of this
normalizer.  Let $Q$ be this normalizer.
\begin{lemma}\label{N unipotent}We have $Q \simeq N \rtimes S$, where
\begin{enumerate}
\item $N$ is unipotent of dimension $r(D-2r) + \frac{1}{2}r(r-1)$.
\item $S \simeq GL(r)\times {\rm SO}(D-2r)$.
\item $P_W$ imbeds into $Q':=N\rtimes \bigg ({\rm SL}(r)\times {\rm SO}(D-2r)\bigg )$.
\end{enumerate}
\end{lemma}
\begin{proof}
Since $U$ is a null space we can find (see Lang \cite{lan}) a null
subspace $U'$ of dimension $r$ of $V$ such that for a suitable basis
$(u_i)$ of $U$ and $(u_i')$ of $U'$ we have $(u_i,u_j') =
\delta_{ij}$. Let $R = (U \oplus U')^\perp$. Then $V=U \oplus U'
\oplus R$ and the quadratic form on $R$ is non-degenerate.  A tedious but straightforward calculation shows that $Q=N\rtimes S$ where $N$ and $S$ are subgroups with the following descriptions. The group $N$ consistes of matrices
of the form
$$\eta(\beta, \sigma) = \left(
                          \begin{array}{ccc}
                           I & -\frac{1}{2}\beta^T \Lambda \beta + \sigma & -\beta^T \Lambda\\
                            0 & I & 0 \\
                            0 & \beta & I \\
                          \end{array}
                        \right)
$$
Here $\beta$ is an arbitrary $(D-2r \times r)$ matrix and , $\sigma$ is an arbitrary $(r \times r)$ skew symmetric matrix. So we have an imbedding of $P_W$ inside $Q$. Since $P_W$ is easily seen to be its own commutator subgroup,
it is now obvious that $P_W$ maps inside the commutator subgroup of $Q$ which is contained in $N\rtimes \bigg ({\rm SL}(r)\times {\rm SO}(R)\bigg )$. This finishes the proof of Lemma
\ref{N unipotent}.
\end{proof}
\begin{lemma}\label{bar L lemma}Suppose ${L}$ acts faithfully on
${k}^3$.  Then ${L}$ acts irreducibly on
${k}^3$.\end{lemma}
\begin{proof}
${L}$ cannot act trivially since it is faithful.  Suppose it is
not irreducible.  then ${k}^3$ has a submodule of dimension $1$
or $2$. By passing to the dual we may assume that $dim  M = 1$. Then
$${L} \subset  {\rm SL}(3)\cap \left\{\left(
                   \begin{array}{ccc}
                     1 & * & * \\
                     0 & * & * \\
                     0 & * & * \\
                   \end{array}
                 \right)\right\}.$$
The left side
has dimension $6$, but the right side has dimension $5$, an impossibility.
\end{proof}
\begin{lemma}\label{null vector lemma}
$U$ has dimension $1$ and is spanned by a null vector $p$.  In
particular $P_W$ fixes $p$.
\end{lemma}
\begin{proof} We assume $r \geq 2$ and show that this
leads to a contradiction. The semisimple part of $P_W$ must have dimension at most the dimension of the semisimple part of $Q'$ and so
$$\frac{(D-2)(D-3)}{2} \leq r^2 -1+ \frac{(D-2r)(D-2r-1)}{2}$$ giving $2D(r-1) \leq 3r^2 + r -4 = (3r+4)(r-1)$ or $2D \leq 3r + 4.$  This gives
$2D \leq 3r + 4 \leq 3 (D/2) + 4$ which cannot happen if $D>8$.  Hence, we need only
consider the possibilities $D=5,6,7,8$.

\vspace {14 pt}

\noindent\textbf{Case $D = 8$:}

\vspace {14 pt}
\noindent $2D \leq 3r + 4 \leq 3 \frac{D}{2} + 4 \Rightarrow 16 \leq
3r + 4 \leq 16 \Rightarrow r=4D/2$. So $V
= U \oplus U'$, $R=0$. Hence $N \simeq k^6$ so
that $$k^6\rtimes \SO(6) \hookrightarrow k^6 \rtimes \SL(4).$$ Since both groups have dimension $21$ the
above map must be an isomorphism. So $\SO(6) \simeq \SL(4)$ which is
impossible since $\SL(4)$ is the two fold cover of $\SO(6)$.  So this
case is ruled out.

\vspace {14 pt}

\noindent\textbf{Case $D = 7$:}

\vspace {14 pt}
\noindent
If $D=7$ then, $2D \leq 3r + 4 \leq 3 (D/2) + 4 \Rightarrow 10 \leq 3r \leq
(21)/2$. No integer exists with these properties.

\vspace{14pt}

\noindent\textbf{Case $D = 6$:}

\vspace{14pt}

\noindent $2D \leq 3r + 4 \leq 3 (D/2) + 4 \Rightarrow 8 \leq
3r \leq 9 \Rightarrow r=3$.  So again $R=0$ and so $P_W=k^4\rtimes {\rm SO}(4)
\hookrightarrow \bar{k}^3 \rtimes {\rm SL}(3)$. Let $\bar{P}$ be the image
of $P_W$ in ${\rm SL}(3)$. Since the ${\rm SO}(4)$ part must map onto a subgroup
of dimension $6$, it follows that $dim(\bar{P})$ has to be $6,7$, or
$8$. If $dim(\bar{P})=8$ (resp.\,$6$) then $\bar{P} = {\rm SL}(3)$ (resp.\,$\bar{P}$ is the image of ${\rm SO}(4)$.) In either, case $\bar{P}$ is
semisimple and has the image of $\bar{k}^4$ as a normal connected
unipotent subgroup.  Hence, the image of $\bar{k}^4$ must be trivial,
which means $\bar{k}^4 \hookrightarrow \bar{k}^3$ which is
impossible. Assume $dim(\bar{P}) = 7$. Let $\bar{T}$ (resp.\,$\bar{L}$) be the
image of ${k}^4$ (resp.\,${\rm SO}(4))$.  Then $\bar{T} \simeq {k}$
and $\bar{L}$ normalizes $\bar{T}$ so that $\bar{L}$ acts trivially
on $\bar{T}$, $\bar{P} = \bar{T} \times \bar{L}$. $\bar{L}$ acts
faithfully on ${k}^3$. By Lemma \ref{bar L lemma}  $\bar{L}$ acts irreducibly on ${k}^3$, so $\bar{T}$ acts
as a scalar, which must be $1$ since $\bar{T} \simeq {k}$.  This
is impossible since $\bar{T}$ must act faithfully.

\vspace{14pt}

\noindent \textbf{Case $D = 5$:}

\vspace{14pt}

\noindent $2D \leq 3r + 4 \leq 3 (D/2) + 4 \Rightarrow 6 \leq
3r \leq (15)/2 \Rightarrow r=2$. Thus $k^3\rtimes {\rm SO}(3)\hookrightarrow k^3\rtimes {\rm SL}(2)$. As in the case $D=8$ we conclude that ${\rm SO}(3)\simeq {\rm SL}(2)$ which is impossible. \end{proof}

\vspace {14 pt}

\noindent
This finishes the proof of Theorem \ref{theorem conjugacy}

\vskip 0.5 true in\noindent
{\mysmall V. S. Varadarajan, Department of Mathematics, UCLA, Los Angeles, CA 90095-1555, USA, {\eightit vsv@math.ucla.edu}} \smallskip\noindent
{\mysmall Jukka Virtanen, Department of Mathematics, UCLA, Los Angeles, CA 90095-1555, USA, {\eightit virtanen@math.ucla.edu}
\medskip

\end{document}